\documentclass[journal,10pt,twocolumn]{IEEEtran}

\usepackage{amsmath,amssymb,amsthm}

\usepackage{mathtools}

\usepackage{booktabs}
\usepackage{graphicx}
\usepackage{subcaption,multicol}

\usepackage{url}

\usepackage{siunitx}

\usepackage{xcolor}

\usepackage{tikz}

\usepackage{microtype}

\usetikzlibrary{positioning,calc,decorations.pathreplacing,arrows.meta}

\IfFileExists{git-hash}%
    {}%
    {}

\newcommand {\R} {{\mathbb{R}}}

\newcommand {\Z} {{\mathbb{Z}}}

\newcommand {\om} {{\omega}}
\newcommand {\la} {{\lambda}}

\newcommand {\indu} {\lambda}
\newcommand {\tmindu} {\indu}

\newcommand {\x} {\boldsymbol{x}}
\newcommand {\hx} {\boldsymbol{\hat{x}}}

\newcommand {\tx} {\tilde{\boldsymbol{x}}}

\newcommand {\y} {\boldsymbol{y}}

\newcommand {\X} {\mathbf{X}}

\newcommand {\tX} {\tilde{\mathbf{X}}}

\newcommand {\Y} {\mathbf{Y}}

\renewcommand {\a} {\boldsymbol{a}}

\newcommand {\h} {\boldsymbol{h}}
\newcommand {\hh} {\boldsymbol{\hat{h}}}

\newcommand {\g} {\boldsymbol{g}}
\newcommand {\hg} {\boldsymbol{\hat{g}}}

\renewcommand {\d} {\boldsymbol{d}}

\newcommand {\G} {\mathbf{G}}

\newcommand {\I} {\mathbf{I}}

\let\realtau\tau
\renewcommand {\tau} {\boldsymbol{\realtau}}

\let\realxi\xi
\renewcommand {\xi} {\boldsymbol{\realxi}}

\let\realvarepsilon\varepsilon
\renewcommand {\varepsilon} {\boldsymbol{\realvarepsilon}}

\newcommand {\tm}{{\mathrm{(t)}}}
\newcommand {\fr}{{\mathrm{(f)}}}

\newcommand {\conv} {\ast}

\newcommand {\tmwav} {\boldsymbol{\psi}}
\newcommand {\tmhwav} {\widehat{\tmwav}}
\newcommand {\tmlow} {\boldsymbol{\phi}}

\newcommand {\tmwavu} {\tmwav}
\newcommand {\tmhwavu} {\tmhwav}
\newcommand {\tmwavd} {\tmwav^{\tm}}

\newcommand {\tfwav} {\boldsymbol{\Psi}}

\newcommand {\tfwavtm} {\tmwavd}

\newcommand {\tfwavfr} {\tmwav^{\fr}}
\newcommand {\tfhwavfr} {\tmhwav^{\fr}}

\newcommand {\gu} {\g}
\newcommand {\hu} {\h}

\newcommand {\hgu} {\hg}
\newcommand {\hhu} {\hh}

\newcommand {\gd} {\g^\tm}
\newcommand {\hd} {\h^\tm}

\renewcommand {\S} {\mathbf{S}}
\newcommand {\tmS} {\S}
\newcommand {\tfS} {\S}

\newcommand {\MEL} {\mathbf{M}}

\newcommand {\imunit} {\mathrm{i}}
\newcommand {\euler} {\mathrm{e}}

\newcommand {\supp} {\mathrm{supp~}}

\newcommand {\diff} {\mathrm{d}}

\newtheorem{theorem}{Theorem}
\newtheorem{lemma}{Lemma}

\title{Joint Time-Frequency Scattering}

\author
{Joakim~And\'en, Vincent Lostanlen, and St\'ephane Mallat
\thanks{This work is supported by the ERC InvariantClass 320959.}
\thanks{J. And\'en is with the Flatiron Institute, New York, NY, USA (e-mail: janden@flatironinstitute.org).}
\thanks{V. Lostanlen is with the Cornell Lab of Ornithology, Cornell University, Ithaca, NY, USA and the Music and Audio Research Laboratory, New York University, New York, NY, USA (e-mail: vincent.lostanlen@nyu.edu).}
\thanks{S. Mallat is with the D\'epartement d'Informatique, Ecole Normale Sup\'erieure, Paris, France, the Coll\`ege de France, Paris, France, and the Flatiron Institute, New York, NY, USA (e-mail: mallat@di.ens.fr).}
}

\begin{document}

\pagestyle{headings}

\maketitle

\begin{abstract}
In time series classification and regression, signals are typically mapped into some intermediate representation used for constructing models.
Since the underlying task is often insensitive to time shifts, these representations are required to be time-shift invariant.
We introduce the joint time-frequency scattering transform, a time-shift invariant representation which characterizes the multiscale energy distribution of a signal in time and frequency.
It is computed through wavelet convolutions and modulus non-linearities and may therefore be implemented as a deep convolutional neural network whose filters are not learned but calculated from wavelets.
We consider the progression from mel-spectrograms to time scattering and joint time-frequency scattering transforms, illustrating the relationship between increased discriminability and refinements of convolutional network architectures.
The suitability of the joint time-frequency scattering transform for time-shift invariant characterization of time series is demonstrated through applications to chirp signals and audio synthesis experiments.
The proposed transform also obtains state-of-the-art results on several audio classification tasks, outperforming time scattering transforms and achieving accuracies comparable to those of fully learned networks.
\end{abstract}


\section{Introduction}
\label{sec:intro}

To extract information from signals, we typically map them into a lower-dimensional representation space where we construct model.
The suitability of these representations depends on their ability to capture signal structure relevant to the task in question, such as classification or regression.
For time series, this often includes the signal's time-frequency geometry.
Figure \ref{fig:scal} shows a time-frequency decomposition, the wavelet transform, applied to two audio recordings.
Both are recordings of a person laughing, so their time-frequency structure is similar, but they also exhibit significant variability.
We would like to construct representations invariant to this type of variability but which adequately capture the time-frequency structure of the signals.

An especially important form of variability is time-shifting (and time-warping deformations).
Indeed, many time series classification and regression tasks are invariant to these transformations.
This work will therefore study representations that are time-shift invariant.

Initial work on audio classification computed representations from time-frequency decompositions, such windowed Fourier transforms.
These include mel-spectrograms, mel-frequency cepstral coefficients (MFCCs) \cite{davis-mermelstein}, modulation spectrograms \cite{hermansky-modulation,thompson-atlas} and correlograms \cite{correlogram,patterson-auditory}.
More recent work employs deep convolutional neural networks---cascades of filter banks alternated with nonlinearities \cite{lecun1998gradient,schmidhuber2015deep,lecun2015deep}.
Filters are learned from data, so each network is adapted to then task, often resulting in excellent performance \cite{graves}.
However, learning typically requires large training sets and extensive computational resources.

This work provides a bridge between traditional time-frequency representations and deep convolutional neural networks.
In particular, we implement the mel-spectrogram as a convolutional network and extend it by adding certain filters to that network which increase its discriminative power while maintaining the amount of time-shift invariance.
These filters are not learned but fixed according to the invariance and discriminability needs of the task.
This simplifies analysis and interpretation of the network.
Fixed filters also reduces the associated computational burden since no training is necessary.

\begin{figure}
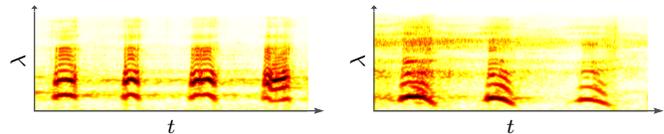

\begin{center}
\input{scal_ex_001_gplt}%
\hspace{0.2cm}%
\input{scal_ex_002_gplt}%
\end{center}
\caption{\label{fig:scal}
The wavelet transform amplitudes, or scalograms, of two recordings as a function of time $t$ and log-frequency $\la$.
Both recordings are of one person laughing.
}
\end{figure}

A convolutional network cascades convolutions, subsampling operators, and pointwise nonlinearities (such as rectifiers) \cite{cnn,relu}.
Its convolution kernels, or filters, are optimized over a training set.
Section \ref{sec:wave} describes how the wavelet transform is computed by a similar cascade of convolutions, but with fixed filters.
A wavelet transform is thus a convolutional network with filters specified by certain time-frequency topology.

To impose time-shift invariance, we compute the modulus of the wavelet transform, known as the scalogram, and average in time.
As shown in Section \ref{sec:mel}, this yields a variant of the popular mel-spectrogram.

Although powerful, mel-spectrograms do not capture large-scale temporal structure, such as amplitude modulation.
In Section \ref{sec:time-scatt}, the time scattering transform extends the mel-spectrogram through multiscale modulation coefficients \cite{stephane,dss}.
Instead of averaging the scalogram, it applies a second wavelet transform in time, takes the modulus, and averages.
This representation is more discriminative and performs well for several classification tasks \cite{dss,emb,talmon2015manifold,sulam2017dynamical}.
Extending the wavelet transform network now lets us implement both mel-spectrograms and time scattering as convolutional networks.

A significant limitation of the time scattering transform is its restriction to convolutions along the time axis.
In other words, its convolutional network is actually a tree, with each node having only a single parent.
A consequence is that time scattering cannot separate signals subjected to time shifts which vary in frequency, which is shown in Section \ref{sec:tf-loss}.
To remedy this, we must capture time and frequency structure jointly.

With this goal in mind, we introduce the joint time-frequency scattering transform.
As described in Section \ref{sec:tf-repr}, it replaces the one-dimensional, channel-by-channel wavelet decomposition of the scalogram by a two-dimensional wavelet transform.
Its construction is inspired by the cortical transform of Shamma et al. \cite{shihab,mesgarani}, which provides neurophysiological models of auditory processing in the mammalian brain.
The corresponding joint scattering network introduces additional filters into the time scattering network, breaking its tree structure and increasing its discriminative power.
To illustrate this, Section \ref{sec:freq-mod} shows how the joint scattering transform captures the chirp rate of frequency-modulated excitations.

The representational power of the proposed transform is further demonstrated in Section \ref{sec:recon} through synthesis experiments.
Here, a signal is synthesized from a target scattering transform by minimizing the distance of its transform to that target.
The resulting synthesized signals show how certain structures which are not captured by the mel-spectrogram and time scattering are better characterized by the joint scattering transform.

Section \ref{sec:classif} concludes by evaluating the joint time-frequency scattering transform on several audio classification tasks.
These include classification of phone segments, musical instruments, and acoustic scenes.
The joint transform outperforms the mel-spectrogram and time scattering while achieving results comparable to, or better than, state-of-the-art convolutional networks.
All figures and tables may be reproduced using software available at \url{http://www.di.ens.fr/data/software/}.

\section{Time-Shift Invariant Representations}
\label{sec:time-freq-repr}

Section \ref{sec:wave} defines the wavelet transform, a representation well suited for time series with multiscale structure.
The modulus of the wavelet transform, known as the scalogram, is averaged in time to yield the time-shift invariant mel-spectrogram, as described in Section \ref{sec:mel}.
Section \ref{sec:time-scatt} reviews the time scattering transform, introduced in And\'en and Mallat \cite{dss}, which extends the invariant mel-spectrogram.
Instead of just averaging the scalogram, it also applies a second wavelet transform, demodulates, and averages the result in time.
These representations are cascades of convolutions and non-linearities and may thus be implemented as deep convolutional networks with fixed filters.

\subsection{Wavelet Transform Filter Bank}
\label{sec:wave}

The wavelet transform of a signal is obtained by convolving it with a set of dilated bandpass filters known as wavelets.
It captures both short, transient structures and long-range oscillations in a localized manner.
In the frequency domain, the ratio between center frequency and bandwidth, the Q factor, is the same for all filters.
These transforms are therefore constant-Q transforms \cite{constant-Q}.
Wavelet filter banks provide good models for cochlear function in mammals \cite{dau,shihab,mesgarani,lewicki} and form the basis for many audio representations \cite{stephane-book}.
The transform may be computed using a multirate filter bank, as has been described in several works \cite{stephane-book,daubechies}.

\begin{figure}
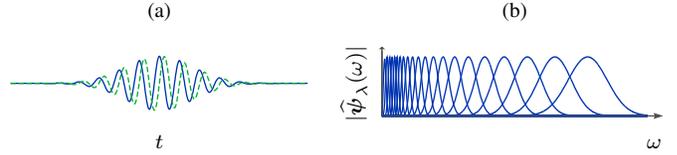

\begin{center}
\input{sample_morlets_1d_time_gplt}%
\hspace{0.2cm}%
\input{sample_morlets_1d_freq_gplt}%
\end{center}
\caption{
\label{fig:morlet}
(a) Real and imaginary parts of a Morlet mother wavelet with $Q = 4$.
(b) The wavelet filters in the frequency domain.
}
\end{figure}

Let $\x(t)$ be a continuous signal for $t \in \R$.
Its Fourier transform is given by
$\hx(\om) =
\int_\R \x(t) \euler^{-2\pi \imunit \omega t}\,\diff t$
for $\om \in \R$.
Following And\'en and Mallat \cite{dss}, we consider a complex analytic wavelet $\tmwavu(t)$ whose Fourier transform $\tmhwav (\om)$ is concentrated in the interval $[2^{-1/Q}, 1]$ for some $Q \ge 1$.
Dilating $\tmwavu(t)$ by factors $2^{-\la}$ now yields the wavelet filter bank
\begin{equation}
\label{eq:dilawave}
\tmwavu_\la (t) = 2^{\la} \tmwavu(2^{\la} t)
\quad \Longleftrightarrow \quad
\tmhwavu_\la (\om) = \tmhwav(2^{-\la} \om)~,
\end{equation}
for $\la \in \R$.
Consequently, $\tmhwavu_\la(\om)$ is concentrated in $[2^{\la-1/Q}, 2^{\la}]$.
This interval has approximate center $2^\la$ and bandwidth $2^\la / Q$.
We therefore need $Q$ filters to cover an octave, independent of frequency.
Since $\tmhwavu_\la(\om)$ is concentrated around $2^\la$, we refer to $\la$ as the wavelet's log-frequency index.

We are typically interested only in structures shorter than some fixed time scale $T$.
In time, $\tmwavu_\la(t)$ has approximate duration $2^{-\la} Q$.
We therefore require $\la$ to satisfy $2^{-\la} Q \le T$.
Unfortunately, certain low frequencies are then not covered by any wavelet.
For audio signals, these frequencies typically contain a small amount of energy and may be safely ignored.
In the following, we instead add a set of constant-bandwidth filters covering these frequencies (see And\'{e}n and Mallat \cite{dss}).

In numerical experiments, we use the Morlet wavelet due to its near-optimal time-frequency localization \cite{stephane-book,dss}.
Figure \ref{fig:morlet} shows a sample Morlet wavelet and its wavelet filter bank.

We now define the continuous wavelet transform of $\x(t)$ as
\begin{equation}
\label{eq:cont-wav}
\x \conv \tmwavu_\la (t)
\end{equation}
for $\la$ such that $2^{-\la} Q \le T$.
It captures the local oscillations of $\x(t)$ at time $t$ and frequency $2^{\la}$ with resolution $2^{-\la} Q$ and $2^{\la} /Q$ in time and frequency, respectively.
In audio applications, we typically set $Q \approx 8$ to better resolve oscillatory components.

\begin{figure}
\begin{tikzpicture}
[
%node distance=0.2cm,
>={Stealth[width=3mm,length=1.5mm]},
font=\small,
signal/.style={draw=none,minimum size=0.2cm},
op/.style={rectangle,draw,minimum height=0.5cm,minimum width=0.95cm,inner sep=0},
missing/.style={draw=none,minimum size=0.2cm},
aggreg/.style=ultra thick
]

% Input

\node [signal] (input) at (0, 0) {$\x$};

% First layer

\node [op] (first-1) [right=0.4cm of input] {$\gu_0$};
\node [op] (first-Q) [below=0.9cm of first-1] {$\gu_{Q-1}$};
\node [missing] (first-missing) at ($(first-1)!0.5!(first-Q)$) {$\vdots$};
\node [op] (first-K) [below=0.9cm of first-Q] {$\hu$};

\foreach \lab in {1, Q}
    \draw [->] (input) -- ++(0.4cm, 0) |- (first-\lab.west);
\draw [->] (input.south) |- (first-K.west);

\node (first-j) [above=0.3cm of first-1] {$j = 1$};

% Second layer

\node [op] (second-K-1) [right=0.9cm of first-K] {$\gu_0$};
\node [op] (second-K-Q) [below=0.9cm of second-K-1] {$\gu_{Q-1}$};
\node [missing] (second-K-missing) at ($(second-K-1)!0.5!(second-K-Q)$) {$\vdots$};
\node [op] (second-K-K) [below=0.9cm of second-K-Q] {$\hu$};

\foreach \labtwo in {1, Q, K}
    \draw [->] (first-K.east) -- ++(0.6cm, 0) |- (second-K-\labtwo.west);

\node (second-j) [above=0cm of second-K-1 |- first-j,anchor=center] {$j = 2$};

% Third layer

\node [op] (third-K-K-1) [right=0.9cm of second-K-K] {$\gu_0$};
\node [op] (third-K-K-Q) [below=1cm of third-K-K-1] {$\gu_{Q-1}$};
\node [missing] (third-K-K-missing) at ($(third-K-K-1)!0.5!(third-K-K-Q)$) {$\vdots$};
\node [op] (third-K-K-K) [below=0.25cm of third-K-K-Q] {$\hu$};

\foreach \labthree in {1, Q, K}
    \draw [->] (second-K-K.east) -- ++(0.6cm, 0) |- (third-K-K-\labthree.west);

\node (third-j) [above=0cm of third-K-K-1 |- first-j,anchor=center] {$j = 3$};

% Fourth layer

\node [op] (fourth-K-K-K-1) [right=0.9cm of third-K-K-K] {$\gu_0$};
\node [op] (fourth-K-K-K-Q) [below=1cm of fourth-K-K-K-1] {$\gu_{Q-1}$};
\node [missing] (fourth-K-K-K-missing) at ($(fourth-K-K-K-1)!0.5!(fourth-K-K-K-Q)$) {$\vdots$};
\node [op] (fourth-K-K-K-K) [below=0.25cm of fourth-K-K-K-Q] {$\hu$};

\foreach \labfour in {1, Q, K}
    \draw [->] (third-K-K-K.east) -- ++(0.6cm, 0) |- (fourth-K-K-K-\labfour.west);

\node (fourth-j) [above=0cm of fourth-K-K-K-1 |- first-j,anchor=center] {$j = 4$};

% Output

\foreach \lab in {first-1, first-Q, second-K-1, second-K-Q, third-K-K-1, third-K-K-Q, fourth-K-K-K-1, fourth-K-K-K-Q, fourth-K-K-K-K}
    \draw [->] (\lab) -- ++(0.9cm,0) node [anchor=west] {};

\end{tikzpicture}
\caption{
\label{fig:wav-conv-net}
Multirate filter bank computing wavelet coefficients for $J = 4$.
Each block corresponds to a filter convolution subsampled by $2$ where a boxed $\hu$ is a low-pass filter and a boxed $\gu_k$ is a band-pass filter.
The depth corresponds to the octave index $j$ while $k = 0, \ldots, Q-1$ is the suboctave index.
}
\end{figure}
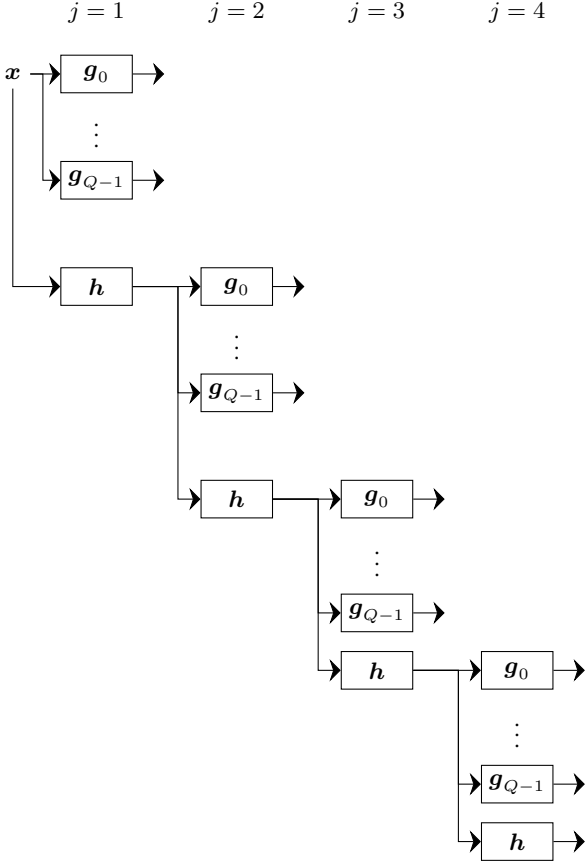

Now let $\x[n]$ be a discrete signal for $n \in \Z$.
Its discrete-time Fourier transform is $\hx(\om) = \sum_{n \in \Z} \x[n]\, \euler^{-2\pi \imunit t \om}$ for $\om \in [-1/2, 1/2]$.
We now define a discrete analog of the continuous wavelet transform \eqref{eq:cont-wav}, implemented as a multirate filter bank.

To achieve this, we consider the multiresolution pyramid obtained by averaging $\x[n]$ at different scales $2^j$.
We initialize the finest scale to $\a_0[n] = \x[n]$.
For $j > 0$, $\a_{j}[n]$ is obtained from $\a_{j-1}[n]$ through convolution by a lowpass filter $\hu[n]$ whose transfer function $\hhu(\om)$ is concentrated in $[-1/4, 1/4]$.
We then subsample by $2$ to obtain
\begin{equation}
\label{eq:hsub}
\a_{j}[n] = \a_{j-1} \conv \hu[2n]~.
\end{equation}
Note that $\a_j[n] = \x \conv \hu_j[2^{j} n]$ for some filter $\hu_j[n]$ defined by
\begin{equation*}
\hhu_j(\om) =
\prod_{p=0}^{j-1} \hhu (2^p \om).
\end{equation*}
As a result, $\hhu_j(\om)$ is concentrated in $[ -2^{-j-1} , 2^{-j-1} ]$ and $\hu_j[n]$ has approximate duration $2^{j+1}$.

The high frequencies of $\a_{j-1}[n]$ lost when convolving with $\h[n]$ are captured by $Q$ bandpass filters $\gu_0[n], \ldots, \gu_{Q-1}[n]$.
Each has a transfer function $\hgu_k(\om)$ concentrated in $[2^{-(k+1)/Q-1}, 2^{-k/Q-1}]$.
After convolving $\a_{j-1}[n]$ with $\gu_k[n]$, the result is subsampled by $2$, yielding
\begin{equation}
\label{eq:gsub}
\d_{j,k}[n] =
\a_{j-1} \conv \gu_k[2n],
\end{equation}
for $j > 0$ and $0 \le k < Q$.
One may verify that
\begin{equation}
\label{eq:discrete-wav}
\d_{j,k}[n] =
\x \conv \gu_{j,k}[2^{j} n],
\end{equation}
where $\hgu_{j,k}(\om) = \hhu_{j-1}(\om)\, \hgu_k (2^{j-1} \om)$.
These filters are concentrated in intervals $[2^{-j-(k+1)/Q}, 2^{-j-k/Q}]$.
In time, they have approximate duration $2^{j} Q$.
Since we are only concerned with local variability below time scale $T$, we require $2^{j} Q \le T$.
This specifies the maximum depth $J = \log_2 (T/Q)$ of the cascade.

Figure \ref{fig:wav-conv-net} illustrates this multirate filterbank cascade.
Each box corresponds to a convolution and subsampling by $2$ according to \eqref{eq:hsub} or \eqref{eq:gsub}.
First, $\x[n]$ is convolved with $\gu_0[n], \ldots, \gu_{Q-1}[n]$ and subsampled to yield the highest octave of bandpass coefficients $\d_{1,0}[n], \ldots, \d_{1,Q-1}[n]$.
Convolving $\x[n]$ with $\hu[n]$ and subsampling provides the remaining low frequencies, and the process is repeated.
As we progress through this cascade, the depth corresponds to the octave index $j$.

Combining the bandpass outputs yields the discrete wavelet transform in \eqref{eq:discrete-wav} for $1 \le j \le J$ and $0 \le k < Q$.
This is similar to the output of the continuous wavelet transform.
Indeed, if we sample a continuous band-limited signal $\x(t)$ at unit intervals, its discrete wavelet transform \eqref{eq:discrete-wav} approximates the continuous transform \eqref{eq:cont-wav} for $\la = -j-k/Q \le -1$ provided that $\hgu_{j,k}(\om) \approx \tmwavu_\la(\om)$.
Given the mother wavelet $\tmwav(t)$, it is possible to construct filters $\hu[n]$ and $\gu_0[n], \ldots, \gu_{Q-1}[n]$ such that this correspondence holds for large $j$ \cite{stephane-book}.
The result is an approximation of the continuous wavelet transform using the convolutional network illustrated in Figure \ref{fig:wav-conv-net}.

\subsection{Mel-Spectrogram}
\label{sec:mel}

The lack of time-shift invariance of the wavelet transform hinders its generalization power for classification.
For most classification tasks, shifting a signal in time does not modify its class.
To reduce variability when constructing models, the signal representation must therefore be made time-shift invariant.
In And\'en and Mallat \cite{dss}, this is achieved by computing the modulus and applying a lowpass filter.
Let us review this construction and study how this may be implemented in a multirate filterbank cascade.

The amplitude of the wavelet transform is the scalogram:
\begin{equation}
\X(t,\lambda) = |\x \conv \tmwavu_\la(t)|.
\end{equation}
Figure \ref{fig:scal} shows two sample scalograms.
Since the wavelets are analytic, applying the complex modulus performs a Hilbert demodulation, capturing the temporal envelope of each subband.
The scalogram $\X(t, \la)$ therefore describes the time-frequency intensity of $\x(t)$ at time $t$ and log-frequency $\la$.

Unfortunately, the scalogram is not time-shift invariant.
Indeed, shifting a signal $\x(t) \mapsto \x_c(t) = \x(t-c)$ also shifts its scalogram $\X(t, \la) \mapsto \X_c(t, \la) = \X(t-c, \la)$.
To ensure invariance, we average in time to obtain
\begin{equation}
\label{eq:mel}
\MEL \x(t, \la) =
\X(\cdot, \la) \conv \tmlow_T(t) =
|\x \conv \tmwavu_\la | \conv \tmlow_T(t),
\end{equation}
where $\tmlow_T(t) = T^{-1} \tmlow(T^{-1} t)$ for some lowpass filter $\tmlow(t)$ of duration $1$, so $\tmlow_T(t)$ has duration $T$.
This is the mel-spectrogram $\MEL \x(t, \la)$ of $\x(t)$.
For $|c| \ll T$, it satisfies $\MEL \x_c(t, \la) \approx \MEL \x(t, \la)$, so it is locally invariant to time-shifts.
The underlying wavelet structure of the mel-spectrogram also ensures stability to time-warping deformations \cite{dss}.

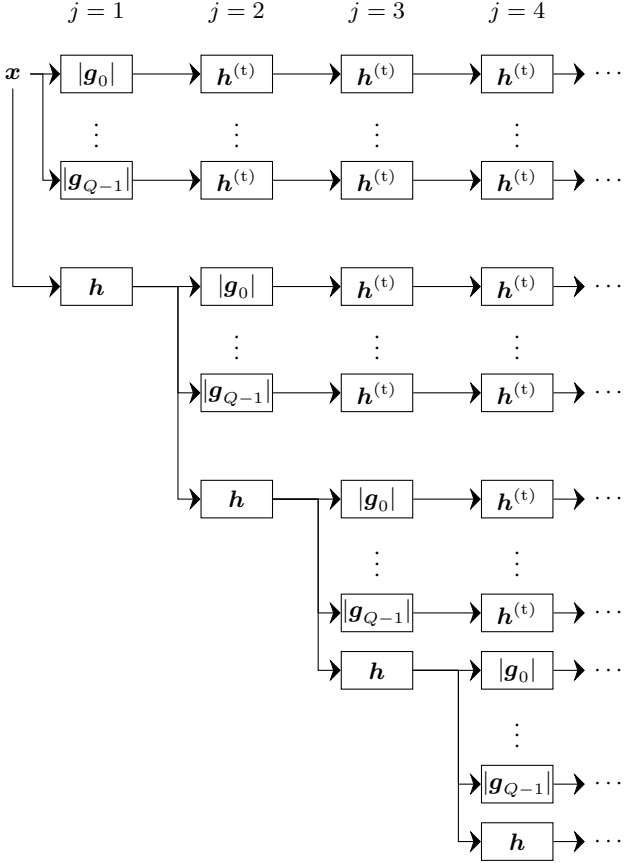
\begin{figure}
\begin{tikzpicture}
[
%node distance=0.2cm,
>={Stealth[width=3mm,length=1.5mm]},
font=\small,
signal/.style={draw=none,minimum size=0.2cm},
op/.style={rectangle,draw,minimum height=0.5cm,minimum width=0.95cm,inner sep=0},
missing/.style={draw=none,minimum size=0.2cm},
aggreg/.style=ultra thick
]

% Input

\node [signal] (input) at (0, 0) {$\x$};

% First layer

\node [op] (first-1) [right=0.4cm of input] {$|\gu_0|$};
\node [op] (first-Q) [below=0.9cm of first-1] {$|\gu_{Q-1}|$};
\node [missing] (first-missing) at ($(first-1)!0.5!(first-Q)$) {$\vdots$};
\node [op] (first-K) [below=0.9cm of first-Q] {$\hu$};

\foreach \lab in {1, Q}
    \draw [->] (input) -- ++(0.4cm, 0) |- (first-\lab.west);
\draw [->] (input.south) |- (first-K.west);

\node (first-j) [above=0.3cm of first-1] {$j = 1$};

% Second layer

\node [op] (second-1) [right=0.9cm of first-1] {$\hd$};
\node [op] (second-Q) [right=0.9cm of first-Q] {$\hd$};
\node [missing] (second-missing) at ($(second-1)!0.5!(second-Q)$) {$\vdots$};

\node [op] (second-K-1) [right=0.9cm of first-K] {$|\gu_0|$};
\node [op] (second-K-Q) [below=0.9cm of second-K-1] {$|\gu_{Q-1}|$};
\node [missing] (second-K-missing) at ($(second-K-1)!0.5!(second-K-Q)$) {$\vdots$};
\node [op] (second-K-K) [below=0.9cm of second-K-Q] {$\hu$};

\foreach \labone in {1, Q}
    \draw [->] (first-\labone.east) -- (second-\labone.west);

\foreach \labtwo in {1, Q, K}
    \draw [->] (first-K.east) -- ++(0.6cm, 0) |- (second-K-\labtwo.west);

\node (second-j) [above=0cm of second-K-1 |- first-j,anchor=center] {$j = 2$};

% Third layer

\node [op] (third-1) [right=0.9cm of second-1] {$\hd$};
\node [op] (third-Q) [right=0.9cm of second-Q] {$\hd$};
\node [missing] (third-missing) at ($(third-1)!0.5!(third-Q)$) {$\vdots$};

\node [op] (third-K-1) [right=0.9cm of second-K-1] {$\hd$};
\node [op] (third-K-Q) [right=0.9cm of second-K-Q] {$\hd$};
\node [missing] (third-K-missing) at ($(third-K-1)!0.5!(third-K-Q)$) {$\vdots$};

\node [op] (third-K-K-1) [right=0.9cm of second-K-K] {$|\gu_0|$};
\node [op] (third-K-K-Q) [below=1cm of third-K-K-1] {$|\gu_{Q-1}|$};
\node [missing] (third-K-K-missing) at ($(third-K-K-1)!0.5!(third-K-K-Q)$) {$\vdots$};
\node [op] (third-K-K-K) [below=0.25cm of third-K-K-Q] {$\hu$};

\foreach \labone in {1, Q, K-1, K-Q}
    \draw [->] (second-\labone.east) -- (third-\labone.west);

\foreach \labthree in {1, Q, K}
    \draw [->] (second-K-K.east) -- ++(0.6cm, 0) |- (third-K-K-\labthree.west);

\node (third-j) [above=0cm of third-K-K-1 |- first-j,anchor=center] {$j = 3$};

% Fourth layer

\node [op] (fourth-1) [right=0.9cm of third-1] {$\hd$};
\node [op] (fourth-Q) [right=0.9cm of third-Q] {$\hd$};
\node [missing] (fourth-missing) at ($(fourth-1)!0.5!(fourth-Q)$) {$\vdots$};

\node [op] (fourth-K-1) [right=0.9cm of third-K-1] {$\hd$};
\node [op] (fourth-K-Q) [right=0.9cm of third-K-Q] {$\hd$};
\node [missing] (fourth-K-missing) at ($(fourth-K-1)!0.5!(fourth-K-Q)$) {$\vdots$};

\node [op] (fourth-K-K-1) [right=0.9cm of third-K-K-1] {$\hd$};
\node [op] (fourth-K-K-Q) [right=0.9cm of third-K-K-Q] {$\hd$};
\node [missing] (fourth-K-K-missing) at ($(fourth-K-K-1)!0.5!(fourth-K-K-Q)$) {$\vdots$};

\node [op] (fourth-K-K-K-1) [right=0.9cm of third-K-K-K] {$|\gu_0|$};
\node [op] (fourth-K-K-K-Q) [below=1cm of fourth-K-K-K-1] {$|\gu_{Q-1}|$};
\node [missing] (fourth-K-K-K-missing) at ($(fourth-K-K-K-1)!0.5!(fourth-K-K-K-Q)$) {$\vdots$};
\node [op] (fourth-K-K-K-K) [below=0.25cm of fourth-K-K-K-Q] {$\hu$};

\foreach \labone in {1, Q, K-1, K-Q, K-K-1, K-K-Q}
    \draw [->] (third-\labone.east) -- (fourth-\labone.west);

\foreach \labfour in {1, Q, K}
    \draw [->] (third-K-K-K.east) -- ++(0.6cm, 0) |- (fourth-K-K-K-\labfour.west);

\node (fourth-j) [above=0cm of fourth-K-K-K-1 |- first-j,anchor=center] {$j = 4$};

% Infinity

\foreach \lab in {1, Q, K-1, K-Q, K-K-1, K-K-Q, K-K-K-1, K-K-K-Q, K-K-K-K}
    \draw [->] (fourth-\lab) -- ++(0.9cm,0) node [anchor=west] {$\cdots$};

\end{tikzpicture}
\caption{
\label{fig:mel-freq-conv-net}
Mel-spectrogram implemented as a convolutional network.
Each $|\gu_k|$ block convolves by a band-pass filter $\gu_k[n]$, computes the modulus, and subsamples by $2$.
Blocks containing $\hu$ or $\hd$ convolve by a low-pass filter and subsample by $2$.
}
\end{figure}

The mel-spectrogram was originally introduced for speech classification \cite{davis-mermelstein} and was motivated by psychoacoustic studies.
It has since found widespread use in various audio classification tasks \cite{logan,somervuo-bird,esc50}.
Traditionally, the mel-spectrogram is computed through frequency averaging of the windowed Fourier transform amplitude, also known as the spectrogram.
However, it has recently been shown that they may be approximated by the time-averaged scalogram coefficients \eqref{eq:mel} \cite{dss,bammer-sampta,dorfler-convnet}.
Note that this formulation makes the time-shift invariance of the mel-spectrogram explicit.
Indeed, the amount of invariance is directly controlled by the duration $T$ of the lowpass filter $\tmlow_T(t)$.
We shall use this wavelet-based variant of the mel-spectrogram in the following.

We now define the discrete mel-spectrogram using the discrete wavelet transform.
The resulting convolutional network is shown in Figure \ref{fig:mel-freq-conv-net}.
Instead of just convolving by $\gu_k[n]$, this network also applies a modulus and subsamples by $2$.
The whole operation is denoted by a boxed $|\gu_k|$.
The result is then passed through a sequence of lowpass filters $\hd[n]$ alternated with subsampling operators, approximating the convolution by $\tmlow_T(t)$.
The output is $JQ+1$ signals of form $|\x \conv \gu_{j,k} | \conv \hd_{J-j} [2^J n]$, where $j$ is the depth at which the modulus was applied.
If the filters are chosen as in Section \ref{sec:wave}, this approximates $\MEL x(t, \la)$ for a bandpass $\x(t)$.

For real $\gu_k[n]$, we may replace the modulus with a rectified linear unit.
Indeed, averaging a rectified bandpass signal approximates its Hilbert envelope, so the result is similar \cite{papoulis}.

\begin{figure}
\begin{tikzpicture}
[
%node distance=0.2cm,
>={Stealth[width=3mm,length=1.5mm]},
font=\small,
signal/.style={draw=none,minimum size=0.2cm},
op/.style={rectangle,draw,minimum height=0.5cm,minimum width=0.95cm,inner sep=0},
missing/.style={draw=none,minimum size=0.2cm},
aggreg/.style=ultra thick
]

% Input

\node [signal] (input) at (0, 0) {$x$};

% First layer

\node [op] (first-1) [right=0.4cm of input] {$|\gu_0|$};
\node [op] (first-Q) [below=2.5cm of first-1] {$|\gu_{Q-1}|$};
\node [missing] (first-missing) at ($(first-1)!0.5!(first-Q)$) {$\vdots$};
\node [op] (first-K) [below=2.0cm of first-Q] {$\hu$};

\foreach \lab in {1, Q}
    \draw [->] (input) -- ++(0.4cm, 0) |- (first-\lab.west);
\draw [->] (input.south) |- (first-K.west);

\node (first-j) [above=0.3cm of first-1] {$j = 1$};

% Second layer

\node [op] (second-1-1) [right=0.9cm of first-1] {$|\gd|$};
\node [op] (second-1-K) [below=0.25cm of second-1-1] {$\hd$};

\node [op] (second-Q-1) [right=0.9cm of first-Q] {$|\gd|$};
\node [op] (second-Q-K) [below=0.25cm of second-Q-1] {$\hd$};

\node [op] (second-K-1) [right=0.9cm of first-K] {$|\gu_0|$};
\node [op] (second-K-Q) [below=1.5cm of second-K-1] {$|\gu_{Q-1}|$};
\node [missing] (second-K-missing) at ($(second-K-1)!0.5!(second-K-Q)$) {$\vdots$};
\node [op] (second-K-K) [below=1.0cm of second-K-Q] {$\hu$};

\foreach \labone in {1, Q}
{
    \foreach \labtwo in {1, K}
        \draw [->] (first-\labone.east) -- ++(0.6cm, 0) |- (second-\labone-\labtwo.west);
}

\foreach \labtwo in {1, Q, K}
    \draw [->] (first-K.east) -- ++(0.6cm, 0) |- (second-K-\labtwo.west);

\node (second-j) [above=0cm of second-K-1 |- first-j,anchor=center] {$j = 2$};

% Third layer

\node [op] (third-1-1) [right=0.9cm of second-1-1] {$\hd$};

\node [op] (third-1-K-1) [right=0.9cm of second-1-K] {$|\gd|$};
\node [op] (third-1-K-K) [below=0.25cm of third-1-K-1] {$\hd$};

\node [op] (third-Q-1) [right=0.9cm of second-Q-1] {$\hd$};
\node [op] (third-Q-K-1) [right=0.9cm of second-Q-K] {$|\gd|$};
\node [op] (third-Q-K-K) [below=0.25cm of third-Q-K-1] {$\hd$};

\node [missing] (third-missing) at ($(third-1-K-K)!0.5!(third-Q-1)$) {$\vdots$};

\node [op] (third-K-1-1) [right=0.9cm of second-K-1] {$|\gd|$};
\node [op] (third-K-1-K) [below=0.25cm of third-K-1-1] {$\hd$};

\node [op] (third-K-Q-1) [right=0.9cm of second-K-Q] {$|\gd|$};
\node [op] (third-K-Q-K) [below=0.25cm of third-K-Q-1] {$\hd$};

\node [op] (third-K-K-1) [right=0.9cm of second-K-K] {$|\gu_0|$};
\node [op] (third-K-K-Q) [below=1cm of third-K-K-1] {$|\gu_{Q-1}|$};
\node [op] (third-K-K-K) [below=0.25cm of third-K-K-Q] {$\hu$};

\node [missing] (third-K-K-missing) at ($(third-K-K-1)!0.5!(third-K-K-Q)$) {$\vdots$};

\foreach \labone in {1, Q}
    \foreach \labtwo in {1}
        \draw [->] (second-\labone-\labtwo.east) -- (third-\labone-\labtwo.west);

\foreach \labone in {1, Q}
    \foreach \labthree in {1, K}
        \draw [->] (second-\labone-K.east) -- ++(0.6cm, 0) |- (third-\labone-K-\labthree.west);

\foreach \labtwo in {1, Q, K}
    \foreach \labthree in {1, K}
        \draw [->] (second-K-\labtwo.east) -- ++(0.6cm, 0) |- (third-K-\labtwo-\labthree.west);

\node (third-j) [above=0cm of third-K-K-1 |- first-j,anchor=center] {$j = 3$};

% Infinity

\foreach \lab in {1-1, 1-K-1, 1-K-K, Q-1, Q-K-1, Q-K-K, K-1-1, K-1-K, K-Q-1, K-Q-K, K-K-1, K-K-Q, K-K-K}
    \draw [->] (third-\lab) -- ++(0.9cm,0) node [anchor=west] {$\cdots$};

\draw[decorate,decoration=brace,thick] ($(third-1-1)+(1.5cm,0.25cm)$) -- node [midway,right] {\small $2$nd order} ($(third-1-K-1)+(1.5cm,-0.25cm)$);

\draw[decorate,decoration=brace,thick] ($(third-Q-1)+(1.5cm,0.25cm)$) -- node [midway,right] {\small $2$nd order} ($(third-Q-K-1)+(1.5cm,-0.25cm)$);

\draw[decorate,decoration=brace,thick] ($(third-K-1-1)+(1.5cm,0.25cm)$) -- node [midway,right] {\small $2$nd order} ($(third-K-1-1)+(1.5cm,-0.25cm)$);

\draw[decorate,decoration=brace,thick] ($(third-K-Q-1)+(1.5cm,0.25cm)$) -- node [midway,right] {\small $2$nd order} ($(third-K-Q-1)+(1.5cm,-0.25cm)$);

\draw[decorate,decoration=brace,thick] ($(third-1-K-K)+(1.5cm,0.25cm)$) -- node [midway,right] {\small $1$st order} ($(third-1-K-K)+(1.5cm,-0.25cm)$);

\draw[decorate,decoration=brace,thick] ($(third-Q-K-K)+(1.5cm,0.25cm)$) -- node [midway,right] {\small $1$st order} ($(third-Q-K-K)+(1.5cm,-0.25cm)$);

\draw[decorate,decoration=brace,thick] ($(third-K-1-K)+(1.5cm,0.25cm)$) -- node [midway,right] {\small $1$st order} ($(third-K-1-K)+(1.5cm,-0.25cm)$);

\draw[decorate,decoration=brace,thick] ($(third-K-Q-K)+(1.5cm,0.25cm)$) -- node [midway,right] {\small $1$st order} ($(third-K-K-Q)+(1.5cm,-0.25cm)$);

\draw[decorate,decoration=brace,thick] ($(third-K-K-K)+(1.5cm,0.25cm)$) -- node [midway,right] {\small $0$th order} ($(third-K-K-K)+(1.5cm,-0.25cm)$);

\end{tikzpicture}
\caption{
\label{fig:scatt-conv-net}
A time scattering network.
Each block with $|\gu_k|$ or $|\gd|$ outputs the modulus of the input convolved with a band-pass filter, subsampled by $2$.
Blocks with $\hu$ and $\hd$ convolves the input with a low-pass filter and subsample by $2$.
}
\end{figure}
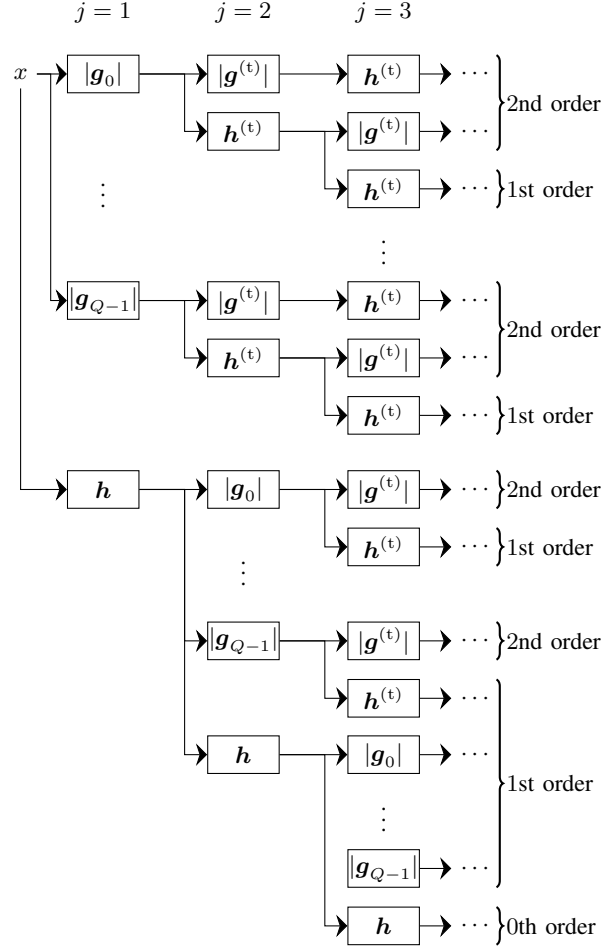

\subsection{Time Scattering}
\label{sec:time-scatt}

The mel-spectrogram discards a large amount of potentially useful information when averaging $\X(t, \la)$ along $t$ in \eqref{eq:mel}, removing any high-frequency structure.
The time scattering transform extends the mel-spectrogram and partially recovers this lost structure while maintaining invariance and stability \cite{stephane,dss}.
This is achieved in And\'en and Mallat \cite{dss} by convolving the scalogram with a second set of wavelets, taking the modulus, and averaging to create second-order time scattering coefficients.
Let us rederive this representation and implement it as a convolutional network extending that of the mel-spectrogram (see Figure \ref{fig:mel-freq-conv-net}).

The first-order time scattering coefficients coincide with the mel-spectrogram $\MEL \x(t, \la)$ and are given by
\begin{equation*}
\tmS_1 \x (t,\la) = \X(\cdot, \la) \conv \tmlow_T(t)~.
\end{equation*}
The lost high frequencies of $\X(t, \la)$ are recovered by convolving with a new set of wavelets, defined from a Morlet mother wavelet $\tmwavd(t)$ by $\tmwavd_\mu (t) = 2^{\mu} \tmwavd(2^{\mu} t)$ for $\mu \in \R$.
Each $\tmwavd_\mu(t)$ has a center frequency of approximately $2^\mu$, so we refer to $\mu$ as their log-frequency.
Unlike their first-order counterparts $\tmwavu_\la(t)$, the second-order wavelets $\tmwavd_\mu(t)$ have $Q = 1$.
As a result, they are better adapted to structures in $\X(t, \la)$, which are less oscillatory and more localized in time compared to those in $\x(t)$.

Convolving $\X(t, \la)$ with these wavelets along $t$, we obtain $\X(\cdot, \la) \conv \tmwavd_\mu(t)$.
To ensure local invariance to translation, we take another modulus and average using $\tmlow_T(t)$, which yields
\begin{equation}
\label{eq:snsd0f8s}
\tmS_2 \x (t,\la,\mu) = |\,|\x \conv \tmwavu_\la| \conv \tmwavd_\mu | \conv \tmlow_T (t)~.
\end{equation}
These are the second-order time scattering coefficients.
They describe the variability of $\X(t, \la)$ along $t$ at frequency $2^\mu$, where $\la$ is the first-order, or acoustic, log-frequency, while $\mu$ is the second-order, or modulation, log-frequency.
As before, we limit ourselves to scales shorter than $T$ by enforcing $2^{-\mu} \le T$.

Concatenating all first- and second-order scattering coefficients $\tmS_1 \x$ and $\tmS_2 \x$ of $\x(t)$ yields the time scattering transform $\tmS \x$.
Higher-order scattering coefficients may be defined \cite{stephane}, but these are of negligible energy \cite{irene} and do not greatly affect classification results \cite{dss}.
The scattering transform exhibits the same amount of time-shift invariance and time-warping stability as the mel-spectrogram described previously \cite{stephane,dss}.
It is more discriminative than the mel-spectrogram, however, since it captures amplitude modulations in $\X(t, \la)$ along $t$.
As a result, the time scattering transform enjoys better performance for classification of audio \cite{dss}, biomedical \cite{emb}, and other types of time series \cite{talmon2015manifold,sulam2017dynamical}.

Other approaches capture temporal structure in the scalogram using Fourier transforms \cite{hermansky-modulation,thompson-atlas} or second-order moments \cite{correlogram,patterson-auditory,mcdermott}.
However, these lack the time-warping stability or noise robustness of the scattering transform \cite{dss,stephane}.

Extending the mel-spectrogram convolutional network of Figure \ref{fig:mel-freq-conv-net}, we define the network of a discrete time scattering transform.
The result is shown in Figure \ref{fig:scatt-conv-net}.
To implement the second-order wavelets $\tmwavd_\mu(t)$, we use the network of Figure \ref{fig:wav-conv-net}, but with a single bandpass filter $\gd[n]$ and a lowpass filter $\hd[n]$.
These are constructed to approximate convolutions with $\tmwavd_\mu(t)$ for $\mu = -j \le -1$ as described in Section \ref{sec:wave}.

As before, $\x[n]$ is first decomposed in the $|\gu_k|$ boxes by convolution with $\gu_0[n], \ldots, \gu_{Q-1}[n]$ followed by modulus and subsampling by $2$.
However, instead of averaging their outputs, they are further convolved with $\gd[n]$ followed by modulus and subsampling, denoted by $|\gd|$ boxes.
These coefficients are then averaged using lowpass filters $\hd[n]$ which alternate with subsampling operators.
This yields the second-order scattering coefficients of $\x[n]$ for the highest octave in $\la$ and the highest octave in $\mu$.
We obtain lower octaves in $\mu$ by applying a sequence of convolutions with $\hd[n]$ alternated with subsampling operators before convolving with $\gd[n]$.
Similarly, lower octaves in $\la$ are obtained by applying a sequence of convolutions by $\h[n]$ and subsampling operators before the decomposition by $\gu_0[n], \ldots, \gu_{Q-1}[n]$.
The outputs of this convolutional network approximate the continuous time scattering transform $\tmS \x$ of $\x(t)$.

\section{Joint Representations in Time and Frequency}
\label{sec:time-freq-scat}

While successfully describing temporal modulation, the time scattering transform fails to capture more sophisticated time-frequency structure, as shown in Section \ref{sec:tf-loss}.
It fails because it decomposes the scalogram as a set of one-dimensional time series.
Section \ref{sec:tf-repr} introduces the joint time-frequency scattering transform, which instead decomposes the scalogram in both time and log-frequency.
Its convolutional network representation introduces connections between nodes in each layer, maintaining the amount of time-shift invariance but increasing its discriminability.
This property is demonstrated in Section \ref{sec:freq-mod}, where we show how the proposed transform accurately captures frequency-modulated excitations.

\subsection{Loss of Time-Frequency Structure}
\label{sec:tf-loss}

\begin{figure}
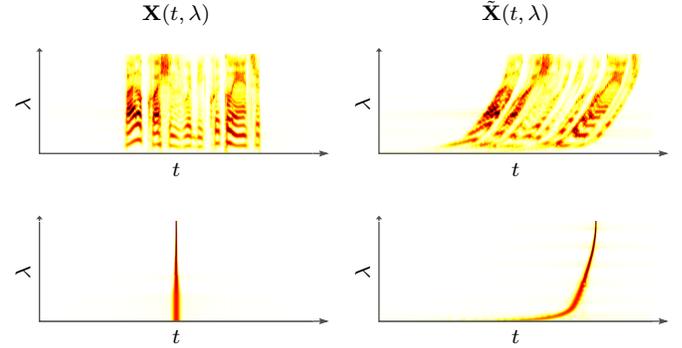

\begin{center}
% Because the top GPLTs "spill out" due to the titles, need a little extra space on top.
\vspace{0.5cm}
\input{questions_original_gplt}%
\hspace{0.2cm}%
\input{questions_deformed_gplt}

\vspace{0.2cm}

\input{dirac_original_gplt}%
\hspace{0.2cm}%
\input{dirac_deformed_gplt}
\end{center}
\caption{\label{fig:two-tones}
Effect of frequency-dependent time-shifts $\tau(\la)$ on scalograms of a speech recording (top) and a Dirac delta function (bottom).
The two columns correspond to the original signal $\x(t)$ and the transformed signal $\tx(t)$, respectively.
}
\end{figure}

The time scattering convolutional network in Figure \ref{fig:scatt-conv-net} has a tree structure; that is, each node only has one parent.
In contrast, a general convolutional network sums contributions from multiple nodes in a layer to produce a node in the next layer.
Due to this tree structure, the time scattering transform is not sensitive to certain time-frequency deformations.

To see this, we suppose that $\x(t)$ is transformed into $\tx(t)$ whose scalogram $\tX(t, \la)$ is an approximate translation of $\X(t, \la)$ by $\tau(\la)$ in each frequency band.
In other words, $\tX(t, \la) \approx \X(t-\tau(\la), \la)$.
Such transformations are illustrated in Figure \ref{fig:two-tones} for a speech signal and a Dirac delta function.
This time-frequency warping misaligns the speech harmonics and transforms the delta function into a chirp.
Although $\x(t)$ differs markedly from $\tx(t)$, this is not detected by time scattering if $|\tau(\la)| \ll T$.
Indeed, the effect of the frequency-varying time shift disappears when averaging by $\tmlow_T(t)$.
Computing the scattering transforms $\tmS \x$ and $\tmS \tx$ for $T$ equal to the signal length yields relative differences $\|\tmS \tx - \tmS \x\|/\|\tmS \x\|$ of $0.07$ and $0.09$ for the speech signal and the delta function, respectively.

Detection of time-frequency warping requires measurement of scalogram variability across frequency.
In particular, the second-order wavelet convolution \eqref{eq:snsd0f8s} in time must be replaced by a convolution in time and log-frequency.

\subsection{Joint Time-Frequency Scattering}
\label{sec:tf-repr}

Existing methods for capturing a signal's time-frequency geometry are not always suitable for classification.
For example, McDermott and Simoncelli \cite{mcdermott} compute higher-order moments of the scalogram across frequencies.
Through synthesis experiments, this representation is shown to provide a good model for audio textures.
However, higher-order moments are not robust to noise, reducing the descriptor's usefulness for classification.

An alternative approach, motivated by neurophysiological studies in the audio cortex of ferrets, is the cortical transform of Shamma et al. \cite{shihab}.
It decomposes the scalogram in both time and log-frequency using two-dimensional Gabor wavelets.
The cortical transform and related representations have brought significant improvements over mel-spectrograms in tasks from speech classification \cite{Schadler2012,Schadler2015} to timbre analysis \cite{mesgarani,kleinschmidt,Siedenburg2016}.
Unfortunately, the lack of time-shift invariance and time-warping stability limits the performance of this approach.

In the following, we adapt the cortical transform within the scattering framework, allowing us to address its invariance and stability.
This also lets us analyze its discriminative power.

We first decompose the scalogram $\X(t, \la)$ using a two-dimensional wavelet transform.
As before, we use Morlet wavelets.
Two-dimensional Morlet wavelets are also used in the two-dimensional scattering transform, which enjoys significant success in natural image classification \cite{joan,laurent}.
In this case, however, the wavelets are obtained by rotating and uniformly scaling a mother wavelet, which is not appropriate for the scalogram.
Indeed, rotation does not preserve the relationship between time and frequency--a rotated scalogram is generally not the scalogram of some other signal.

\begin{figure}
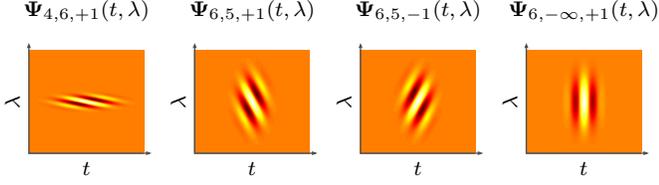

\begin{center}
\vspace{0.5cm}
\input{sample_morlets_2d_001_gplt}%
\input{sample_morlets_2d_002_gplt}%
\input{sample_morlets_2d_003_gplt}%
\input{sample_morlets_2d_004_gplt}%
\end{center}
\caption{
\label{fig:sample-psis}
Real parts of four time-frequency wavelets $\tfwav_{\mu,\ell,s}(t, \la)$.
White-yellow is negative, orange is zero, and red-black is positive.
}
\end{figure}

We instead define our wavelets separably, with independent scaling along time and log-frequency.
The time-frequency mother wavelet $\tfwav(t, \la) = \tfwavtm (t) \, \tfwavfr(\la)$ is the product of two one-dimensional functions in time and log-frequency.
Both the time $\tfwavtm(t)$ and the frequency $\tfwavfr(\la)$ wavelets are Morlet wavelets with $Q = 1$.
Dilating by $2^{-\mu}$ along $t$, dilating by $2^{-\ell}$ along $\la$, and reflecting according to $s$ yields the wavelet
\begin{equation}
\label{eq:tfwav-def}
\tfwav_{\mu,\ell,s} (t,\la) = 2^{\mu+\ell} \, \tfwavtm(2^{\mu} t) \, \tfwavfr(s 2^{\ell} \lambda)~,
\end{equation}
where the spin $s = \pm 1$ specifies the oscillation direction (up or down).
The frequency of the wavelet along $t$ is $2^\mu$, so $\mu$ is the log-frequency of $\tfwav_{\mu,\ell,s}(t, \la)$.
Its frequency along $\la$ is $2^\ell$, so we refer to it as a ``quefrency.''
Consequently, $\ell$ is the ``log-quefrency'' of $\tfwav_{\mu,\ell,s}(t, \la)$.

As before, $\mu$ satisfies $2^{-\mu} \le T$.
Along $\la$, we fix some maximum log-frequency scale $F$, measured in octaves, and let $2^{-\ell} \le F$.
At this maximum scale, we include a lowpass filter to capture average structure along $\la$.
Specifically, we set
\begin{equation}
\label{eq:tfwavlow-def}
\tfwav_{\mu,-\infty,+1} (t,\la) =
2^{\mu} \tfwavtm(2^{\mu} t) \, \tmlow_F (\lambda)~.
\end{equation}
Note that these are only defined for $s = +1$.
Figure \ref{fig:sample-psis} shows a few sample two-dimensional wavelets $\tfwav_{\mu,\ell,s}(t, \la)$.

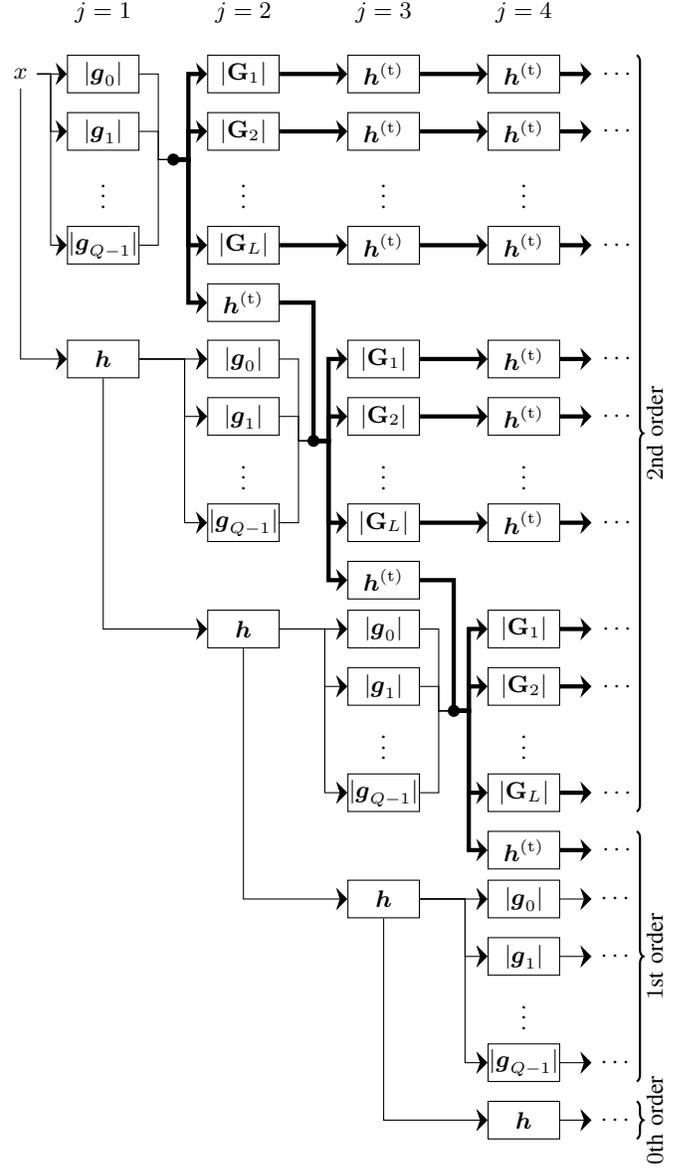
\begin{figure}
\begin{tikzpicture}
[
%node distance=0.2cm,
>={Stealth[width=3mm,length=1.5mm]},
font=\small,
signal/.style={draw=none,minimum size=0.2cm},
op/.style={rectangle,draw,minimum height=0.5cm,minimum width=0.95cm,inner sep=0},
missing/.style={draw=none,minimum size=0.2cm},
aggreg/.style=ultra thick
]

% Input

\node [signal] (input) at (0, 0) {$x$};

% First layer

\node [op] (first-1) [right=0.4cm of input] {$|\gu_0|$};
\node [op] (first-2) [below=0.25cm of first-1] {$|\gu_1|$};
\node [op] (first-Q) [below=1cm of first-2] {$|\gu_{Q-1}|$};
\node [missing] (first-missing) at ($(first-2)!0.5!(first-Q)$) {$\vdots$};
\node [op] (first-K) [below=1cm of first-Q] {$\hu$};

\foreach \lab in {1, 2, Q}
    \draw [->] (input) -- ++(0.4cm, 0) |- (first-\lab.west);
\draw [->] (input.south) |- (first-K.west);

\node (first-j) [above=0.3cm of first-1] {$j = 1$};

% Second layer

\node [op] (second-1-1) [right=0.9cm of first-1] {$|\G_{1}|$};
\node [op] (second-1-2) [right=0.9cm of first-2] {$|\G_{2}|$};
\node [op] (second-1-Q) [right=0.9cm of first-Q] {$|\G_{L}|$};
\node [missing] (second-1-missing) at ($(second-1-2)!0.5!(second-1-Q)$) {$\vdots$};
\node [op] (second-1-K) [below=0.25cm of second-1-Q] {$\hd$};

\node [op] (second-K-1) [right=0.9cm of first-K] {$|\gu_0|$};
\node [op] (second-K-2) [below=0.25cm of second-K-1] {$|\gu_1|$};
\node [op] (second-K-Q) [below=0.9cm of second-K-2] {$|\gu_{Q-1}|$};
\node [missing] (second-K-missing) at ($(second-K-2)!0.5!(second-K-Q)$) {$\vdots$};
\node [op] (second-K-K) [below=0.9cm of second-K-Q] {$\hu$};

\draw (barycentric cs:first-1=1,first-Q=1,second-1-1=1,second-1-Q=1) node (first-mux) {};

\foreach \labone in {1, 2, Q}
    \draw (first-\labone.east) -|
          ($(first-mux.center)+(-0.2cm,0)$) --
          (first-mux.center);

\foreach \labtwo in {1, 2, Q, K}
    \draw [->,aggreg] (first-mux.center) --
               ($(first-mux.center)+(+0.2cm,0)$) |-
               (second-1-\labtwo.west);

\fill (first-mux) circle [radius=0.08cm];

\foreach \labtwo in {1, 2, Q}
    \draw [->] (first-K.east) -- ++(0.6cm, 0) |- (second-K-\labtwo.west);
\draw [->] (first-K.south) |- (second-K-K.west);

\node (second-j) [above=0cm of second-K-1 |- first-j,anchor=center] {$j = 2$};

% Third layer

\node [op] (third-1-1) [right=0.9cm of second-1-1] {$\hd$};
\node [op] (third-1-2) [right=0.9cm of second-1-2] {$\hd$};
\node [op] (third-1-Q) [right=0.9cm of second-1-Q] {$\hd$};
\node [missing] (third-1-missing) at ($(third-1-2)!0.5!(third-1-Q)$) {$\vdots$};

\node [op] (third-K-1) [right=0.9cm of second-K-1] {$|\G_{1}|$};
\node [op] (third-K-2) [right=0.9cm of second-K-2] {$|\G_{2}|$};
\node [op] (third-K-Q) [right=0.9cm of second-K-Q] {$|\G_{L}|$};
\node [missing] (third-K-missing) at ($(third-K-2)!0.5!(third-K-Q)$) {$\vdots$};
\node [op] (third-K-K) [below=0.25cm of third-K-Q] {$\hd$};

\node [op] (third-K-K-1) [right=0.9cm of second-K-K] {$|\gu_0|$};
\node [op] (third-K-K-2) [below=0.25cm of third-K-K-1] {$|\gu_1|$};
\node [op] (third-K-K-Q) [below=0.9cm of third-K-K-2] {$|\gu_{Q-1}|$};
\node [missing] (third-K-K-missing) at ($(third-K-K-2)!0.5!(third-K-K-Q)$) {$\vdots$};
\node [op] (third-K-K-K) [below=0.9cm of third-K-K-Q] {$\hu$};

\foreach \lab in {1, 2, Q}
    \draw [->,aggreg] (second-1-\lab.east) -- (third-1-\lab.west);

\draw (barycentric cs:second-K-1=1,second-K-Q=1,third-K-1=1,third-K-Q=1) node (second-mux) {};

\foreach \labone in {K-1, K-2, K-Q}
    \draw (second-\labone.east) -|
          ($(second-mux.center)+(-0.2cm,0)$) --
          (second-mux.center);
\draw [aggreg] (second-1-K.east) -|
              (second-mux.center);
\foreach \labtwo in {1, 2, Q, K}
    \draw [->,aggreg] (second-mux.center) --
                     ($(second-mux.center)+(+0.2cm,0)$) |-
                     (third-K-\labtwo.west);

\fill (second-mux) circle [radius=0.08cm];

\foreach \labtwo in {1, 2, Q}
    \draw [->] (second-K-K.east) -- ++(0.6cm, 0) |- (third-K-K-\labtwo.west);
\draw [->] (second-K-K.south) |- (third-K-K-K.west);

\node (third-j) [above=0cm of third-K-K-1 |- first-j,anchor=center] {$j = 3$};

% Fourth layer

\node [op] (fourth-1-1) [right=0.9cm of third-1-1] {$\hd$};
\node [op] (fourth-1-2) [right=0.9cm of third-1-2] {$\hd$};
\node [op] (fourth-1-Q) [right=0.9cm of third-1-Q] {$\hd$};
\node [missing] (fourth-1-missing) at ($(fourth-1-2)!0.5!(fourth-1-Q)$) {$\vdots$};

\node [op] (fourth-K-1) [right=0.9cm of third-K-1] {$\hd$};
\node [op] (fourth-K-2) [right=0.9cm of third-K-2] {$\hd$};
\node [op] (fourth-K-Q) [right=0.9cm of third-K-Q] {$\hd$};
\node [missing] (fourth-K-missing) at ($(fourth-K-2)!0.5!(fourth-K-Q)$) {$\vdots$};

\node [op] (fourth-K-K-1) [right=0.9cm of third-K-K-1] {$|\G_1|$};
\node [op] (fourth-K-K-2) [right=0.9cm of third-K-K-2] {$|\G_2|$};
\node [op] (fourth-K-K-Q) [right=0.9cm of third-K-K-Q] {$|\G_L|$};
\node [missing] (fourth-K-K-missing) at ($(fourth-K-K-2)!0.5!(fourth-K-K-Q)$) {$\vdots$};
\node [op] (fourth-K-K-K) [below=0.25cm of fourth-K-K-Q] {$\hd$};

\node [op] (fourth-K-K-K-1) [right=0.9cm of third-K-K-K] {$|\gu_0|$};
\node [op] (fourth-K-K-K-2) [below=0.25cm of fourth-K-K-K-1] {$|\gu_1|$};
\node [op] (fourth-K-K-K-Q) [below=0.9cm of fourth-K-K-K-2] {$|\gu_{Q-1}|$};
\node [missing] (fourth-K-K-K-missing) at ($(fourth-K-K-K-2)!0.5!(fourth-K-K-K-Q)$) {$\vdots$};
\node [op] (fourth-K-K-K-K) [below=0.25cm of fourth-K-K-K-Q] {$\hu$};

\foreach \lab in {1-1, 1-2, 1-Q, K-1, K-2, K-Q}
    \draw [aggreg,->] (third-\lab.east) -- (fourth-\lab.west);

\draw (barycentric cs:third-K-K-1=1,third-K-K-Q=1,fourth-K-K-1=1,fourth-K-K-Q=1) node (third-mux) {};

\foreach \labone in {K-1, K-2, K-Q}
    \draw (third-K-\labone.east) -|
          ($(third-mux.center)+(-0.2cm,0)$) --
          (third-mux.center);
\draw [aggreg] (third-K-K.east) -|
              (third-mux.center);
\foreach \labtwo in {1, 2, Q, K}
    \draw [->,aggreg] (third-mux.center) --
                     ($(third-mux.center)+(+0.2cm,0)$) |-
                     (fourth-K-K-\labtwo.west);

\fill (third-mux) circle [radius=0.08cm];

\foreach \labtwo in {1, 2, Q}
    \draw [->] (third-K-K-K.east) -- ++(0.6cm, 0) |- (fourth-K-K-K-\labtwo.west);
\draw [->] (third-K-K-K.south) |- (fourth-K-K-K-K.west);

\node (fourth-j) [above=0cm of fourth-K-K-K-1 |- first-j,anchor=center] {$j = 4$};

% Infinity

\foreach \lab in {1-1, 1-2, 1-Q, K-1, K-2, K-Q, K-K-1, K-K-2, K-K-Q, K-K-K}
    \draw [aggreg,->] (fourth-\lab) -- ++(0.9cm,0) node [anchor=west] {$\cdots$};
\foreach \lab in {K-K-K-1, K-K-K-2, K-K-K-Q, K-K-K-K}
    \draw [->] (fourth-\lab) -- ++(0.9cm,0) node [anchor=west] {$\cdots$};

\draw[decorate,decoration=brace,thick] ($(fourth-1-1)+(1.5cm,0.25cm)$) -- node [midway,below,rotate=90] {\small $2$nd order} ($(fourth-K-K-Q)+(1.5cm,-0.25cm)$);

\draw[decorate,decoration=brace,thick] ($(fourth-K-K-K)+(1.5cm,0.25cm)$) -- node [midway,below,rotate=90] {\small $1$st order} ($(fourth-K-K-K-Q)+(1.5cm,-0.25cm)$);

\draw[decorate,decoration=brace,thick] ($(fourth-K-K-K-K)+(1.5cm,0.25cm)$) -- node [midway,below,rotate=90] {\small $0$th order} ($(fourth-K-K-K-K)+(1.5cm,-0.25cm)$);

\end{tikzpicture}
\caption{
\label{fig:joint-conv-net}
A joint time-frequency scattering network.
Each $|\gu_k|$ block convolves a one-dimensional signal with the band-pass filter $\gu_k$ and outputs its modulus.
The outputs are aggregated into two-dimensional arrays shown by thick lines.
A $|\G_\ell|$ block convolves a two-dimensional array with the band-pass filter $\G_\ell$ and outputs its modulus.
The $\hu$ and $\hd$ blocks convolve only in time.
All blocks subsample their output in time by $2$.
}
\end{figure}

The two-dimensional wavelet transform of the scalogram $\X(t,\la)$ computes convolutions $\X \conv \tfwav_{\mu,\ell,s}(t,\la)$.
It captures the joint variability of $\X(t,\la)$ at log-frequency $\mu$ and log-quefrency $\ell$ with spin $s$.
To ensure time-shift invariance and time-warping stability, we take the complex modulus and average, obtaining the second-order joint time-frequency scattering coefficients
\begin{equation}
\label{eq:joint-scatt}
\tfS_2 \x(t,\la,\mu,\ell,s) =
| \X \conv \tfwav_{\mu,\ell,s} (\cdot,\tmindu)| \conv \tmlow_T (t).
\end{equation}
These coefficients describe the time-frequency geometry of $\x(t)$ at time $t$ and log-frequency $\la$.
They retain the time-shift invariance and time-warping stability of the second-order time scattering coefficients, but with increased discriminative power.

Concatenating the first-order time scattering coefficients $\tmS_1 \x$ and the second-order time-frequency scattering coefficients $\tfS_2 \x$ yields the complete joint time-frequency scattering transform $\tfS \x$ of $\x(t)$.
As for time scattering, we may define higher-order coefficients, but these are often of limited use for classification.
For each $t$, there are $O(Q \log_2 T)$ first-order coefficients and $O(Q (\log_2 T)^2 \log_2 F)$ second-order coefficients.

We now define a convolutional network to provide a discrete implementation of the joint scattering transform.
In the time scattering network (see Figure \ref{fig:scatt-conv-net}), we approximate the convolution of $\X(t, \la)$ with $\tmwavd_\mu(t)$ along $t$ by cascading discrete filters $\hd[n]$ and $\gd[n]$, alternated with subsampling operators.
The joint transform network incorporates additional filters along the discrete log-frequency $m = Q \la = -jQ-k \in \Z$, where, as before, $j$ and $k$ are the octave and subband indices of $\la$.

In a given layer, the modulus bandpass outputs of the previous layer are arranged along time $n$ and log-frequency $m$ into a two-dimensional array.
This array is then filtered along $m$ by different filters $2^\ell \tfwavfr(s 2^\ell m/Q)$.
It is also filtered by $\tmlow_F(m/Q)$ to account for $\ell = -\infty$.
The sampling interval of the filters is $1/Q$, since this is the spacing of the discretized log-frequencies $\la = m/Q$.
Each frequency-filtered array is then filtered by $\gd[n]$ along $n$.

Combining these into two-dimensional filters, we get
\begin{align*}
& \G_{\ell,s} [n,m] = \gd[n] \, 2^\ell \tfwavfr(s 2^\ell m/Q) \\
& \G_{-\infty,+1} [n,m] = \gd[n] \, \tmlow_F (m/Q)~,
\end{align*}
where $\ell \in \Z$ such that $-\log_2 F \le \ell \le \log_2 Q$ (to ensure that $1/Q \le 2^{-\ell} \le F$) and $s = \pm 1$.
Abusing notation slightly, we renumber this set of discrete filters as $\G_1[n,m], \ldots, \G_{L}[n,m]$.
These filters capture all log-quefrencies along $\la$, but only high frequencies along $n$.
The missing low frequencies are absorbed by $\hd[n]$, which averages along $n$, leaving $m$ intact.

Using these filters, we construct the convolutional network shown in Figure \ref{fig:joint-conv-net}, extending the time scattering network of Figure \ref{fig:scatt-conv-net}.
Small circles denote aggregation of multiple time series into a two-dimensional array, while the arrays themselves are thick lines.
We denote by a boxed $|\G_{\ell}|$ convolution with $\G_{\ell}[n,m]$ for $\ell = 1, \ldots, L$, followed by a complex modulus and subsampling by $2$ along $n$.
Similarly, a boxed $\hd$ denotes lowpass filtering along $n$ by $\hd[n]$ followed by subsampling.

Starting with a signal $\x[n]$, we first compute its decomposition using the first-order blocks $|\gu_0|, \ldots, |\gu_{Q-1}|$, extracting the highest octave of the signal.
We then combine these into a two-dimensional array which is decomposed by $|\G_1|, \ldots, |\G_L|$.
The outputs of $|\G_1|, \ldots, |\G_L|$ are then forwarded to a succession of $\hd$ blocks which implement the averaging by $\tmlow_T[n]$.
The original array is also decomposed by $\hd$, and the result is concatenated to the first-order outputs of the second layer (that is, the second octave of the original signal).
We then repeat the process on this array.
As before, an appropriate choice of $\gd[n]$ and $\hd[n]$ ensures that the network accurately approximates the continuous joint scattering transform.

The important difference between this network and the time scattering network is the presence of within-layer connections.
These break the tree structure, increasing discriminative power through better characterization of time-frequency geometry.
Returning to the frequency-warped signals of Figure \ref{fig:two-tones}, the joint network separates the original and transformed signals, with $\|\tfS \tx - \tfS \x\|/\|\tfS \x\|$ of $0.41$ and $0.90$, compared to $0.07$ and $0.09$ for time scattering.
This network therefore has same time-shift invariance as time scattering, but with better discriminability.

Note that this increased discriminative power may not always be desirable.
For example, frequency-dependent time-shifts (as shown in Figure \ref{fig:two-tones}) or similar transformations may not be relevant for the classification task.
In this case, replacing the time scattering transform with a joint time-frequency scattering transform would needlessly increase the number of model parameters, potentially requiring more training data to train an accurate classifier.
On the other hand, the high-quefrency joint coefficients approximate the standard second-order time scattering coefficients.
As a result, the types of structures captured by the time scattering transform are equally well characterized by the joint transform, so little discriminative power is lost by replacing the former by the latter.

The invariance and discriminability properties of the transform are controlled by three parameters: $Q$, $T$, and $F$. The number of wavelets per octave, $Q$, depends on the time-frequency localization of the input signal. For example, if the signal is highly oscillatory (that is, well-localized in frequency, but not necessarily in time), a higher value for $Q$ is appropriate. This is the case in audio, but not necessarily for biomedical or geophysical time series, which are more localized in time.

The averaging scale $T$ controls the maximum length of the signal structure captured by the transform. In other words, if the relevant structure in a classification problem occurs at very small scales, $T$ should be kept small. This is the case in phone segment classification (see Section \ref{sec:timit}), where the object of interest, the phone, is of short duration. For other signals, such as musical instrument recordings (see Section \ref{sec:medleydb}), there are relevant structures at larger scales. The $T$ parameter also controls the length of the lowpass filter $\tmlow_T(t)$ and therefore determines the amount of desired time-shift invariance.

The frequency scale $F$ has a similar role, controlling the maximum frequency extent of the signal structure captured by the transform. If we expect relevant structures to spread out over several octaves, a large value for $F$ is needed. This is the case for speech signals, where plosive phones occupy a large part of the frequency domain. For other signals, such environmental sounds, relevant structures may be confined within an octave, so a small $F$ is more appropriate.

The output of a scattering network may be used as input to another convolutional network whose filters are subsequently optimized for some classification task.
This yields a large convolutional network taking raw waveforms as input and whose first few layers are fixed.
By fixing certain layers, the network has fewer parameters to optimize and could then be trained using less data.
Previous work training convolutional networks on raw waveforms have yielded mel-like filters in the first few layers \cite{tuske2014acoustic,sainath2015learning}, providing some support for this idea.
Other attempts at explicitly incorporating wavelets into convolutional network architectures have also demonstrated the viability of the approach \cite{fujieda2018wavelet,oyallon2018compressing,shi20173d,liao2017wavelet}.
In addition, the success of transfer learning \cite{hamel2013transfer,vandenoord2014transfer,aytar,arandjelovic-zisserman} suggests that there exist certain universal representations which perform well for a wide range of tasks.
The joint scattering network provides a way to construct such a representation while enforcing certain time-shift invariance and time-frequency discriminability conditions.

\subsection{Frequency Modulation}
\label{sec:freq-mod}

The above construction is similar to that of traditional convolutional networks except that filters are not learned from data.
These filters provide the time-shift invariance and time-warping stability of the time scattering transform, but the joint transform is also more discriminative.
To illustrate this, we show how the joint time-frequency scattering transform captures frequency modulation structure ignored by time scattering.

Let $\x(t) = \exp(2\pi \imunit\, \xi(t))$ be a frequency-modulated excitation with instantaneous phase $\xi(t)$.
At time $t$, its instantaneous frequency is given by $\xi'(t)$, while the relative change in this frequency, the (relative) chirp rate, is $\xi''(t)/\xi'(t)$.
Frequency modulation occurs in a variety of signals, such as speech, animal calls, music and radar signals \cite{richards2005fundamentals}.

We now consider a particular case of frequency modulation: the exponential chirp.
Here $\xi(t) = 2^{\alpha t}$, so it has instantaneous frequency $\xi'(t) = \alpha \log (2) 2^{\alpha t}$ and constant chirp rate $\xi''(t)/\xi'(t) = \alpha \log(2)$.
We note that an arbitrary frequency-modulated excitation may be locally approximated by an exponential chirp by setting $\alpha = (\log 2)^{-1} \, \xi''(t)/\xi'(t)$.

For exponential chirps, we have the following result.
\begin{theorem}
\label{thm:fm}
Let $\tmwavu_\la(t)$ and $\tfwav_{\mu,\ell,s}(t, \la)$ be defined as in \eqref{eq:dilawave} and \eqref{eq:tfwav-def}.
We require that $\tmwav(t)$ have compact support, that $\|\tmwav\|_\infty$, $\|\tmwav'\|_1$, $\int_\R |u| |\tmhwav (2^u)| \,\diff u$, and $\|{\tfwavtm}'\|_\infty$ are bounded, and that $\supp \tfwavfr(\la) \subset [-A, A]$ for some $A > 0$.
Further, we assume that $\tfwavtm(t)$ is the product of a positive envelope $|\tfwavtm|(t)$ and  $\exp(2\pi \imunit \, t)$.
Let $\x(t) = \exp(2\pi \imunit \, 2^{\alpha t})$ for some $\alpha \in \R$.
The joint scattering transform \eqref{eq:joint-scatt} then satisfies
\begin{align*}
&\tfS_2 \x(t, \la, \mu, \ell, s) \\
&\quad = \frac{c_0 E(t, \la)}{\alpha} \, \left| \tfhwavfr\left(-\frac{s 2^{\mu-\ell}}{\alpha}\right)\right| + \varepsilon(t, \la, \mu, \ell, s)~,
\end{align*}
where
\begin{equation*}
E(t, \la) = |\tfwavtm_\mu| \conv \phi_T\left( t - \frac{\la}{\alpha} + \frac{\log\log 2^\alpha}{\log 2^\alpha} \right)~,
\end{equation*}
\begin{equation*}
\|\varepsilon\|_\infty < C \left(|\alpha| 2^{-\la + 2^{-\ell} A} + 2^{2\mu} |\alpha|^{-2} + 2^{2\mu-\ell} |\alpha|^{-2}\right)~,
\end{equation*}
for $C > 0$ depending only on $\tmwav(t)$, $\tfwavtm(t)$, and $\tfwavfr(\la)$, and $c_0 = \int_\R |\tmhwav(2^u)| \,\diff u$.
\end{theorem}
The proof is given in Appendix \ref{sec:appendix}.
The result relies on approximating $\X(t, \la)$ by $|\tmhwav(\log (2^\alpha) 2^{-\la+\alpha t})|$.
Since $|\tmhwav(\om)|$ is maximized at $\om = 1$, this forms a ridge $\la = \alpha t$ with slope $\alpha$, as illustrated in Figure \ref{fig:chirp-model}(a,b).
In the joint transform, this ridge only activates certain second-order wavelets $\tfwav_{\mu,\ell,s}(t, \la)$.
Indeed, only wavelets whose slope $-s2^{\mu-\ell}$ aligns with $\la = \alpha t$ yield large coefficients.
Taking the complex modulus and averaging in time preserves this slope information.

\begin{figure}
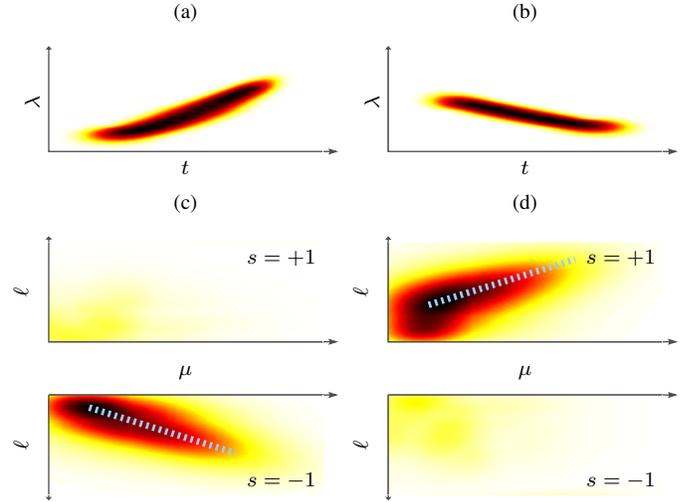

\vspace{0.5cm}
\input{chirp_001_scal_gplt}%
\hspace{0.2cm}%
\input{chirp_002_scal_gplt}%

\vspace{0.5cm}
\input{chirp_001_joint_pos_gplt}%
\hspace{0.2cm}%
\input{chirp_002_joint_pos_gplt}%

\input{chirp_001_joint_neg_gplt}%
\hspace{0.2cm}%
\input{chirp_002_joint_neg_gplt}%
\caption{
\label{fig:chirp-model}
Scalograms of two exponential chirps with chirp rates (a) $\alpha = 4$ and (b) $\alpha = -2$.
(c, d) Corresponding second-order joint time-frequency scattering coefficients $\tfS_2 \x(t, \la, \mu, \ell, s)$ for fixed $t$ and $\la$. The dotted lines satisfy $s2^{\mu-\ell} = -\alpha$.
}
\end{figure}

Let us consider another chirp $\tx(t) = \exp(2\pi \imunit \, 2^{\tilde{\alpha} t})$.
We may obtain $\tx(t)$ from $\x(t)$ using a frequency-dependent time-shift of its scalogram $\X(t, \la)$ as in Section \ref{sec:tf-loss}.
Here, we take
\begin{equation*}
\tau(\la) = \la\left(\frac{1}{\alpha} - \frac{1}{\tilde\alpha}\right) - \frac{\log\log 2^\alpha}{\log 2^\alpha} + \frac{\log\log 2^{\tilde \alpha}}{\log 2^{\tilde \alpha}}~.
\end{equation*}
As we saw in Section \ref{sec:tf-loss}, the time scattering transform is not sensitive to such changes.
In other words, the scattering transform discards information on slope, rendering it unsuitable for describing frequency modulation.
The same applies to related representations which also decompose each subband of $\X(t, \la)$ separately, such as mel-spectrograms, MFCCs, and modulation spectrograms.
This information loss is fundamentally due to the tree structure of their convolutional networks.

Theorem \ref{thm:fm} states that, for fixed $t$ and $\la$, $\tfS_2 \x(t, \la, \mu, \ell, s)$ is approximately proportional to $|\tfhwavfr(-s 2^{\mu-\ell} \alpha^{-1})|$.
Since $\tfhwavfr$ is concentrated around frequency $1$, this is maximized for $-s 2^{\mu-\ell} \alpha^{-1} = 1$.
In other words, a ridge is present along $s 2^{\mu-\ell} = -\alpha$.
Frequency modulation structure in the form of the chirp rate $\alpha$, is thus encoded in the second-order joint time-frequency scattering coefficients.
Consequently, they are sensitive to frequency-dependent time-shifts $\X(t, \la) \mapsto \X(t-\tau(\la), \la)$ even when $|\tau(\la)| \ll T$, since these change $\alpha$.

Figure \ref{fig:chirp-model}(c,d) displays a subset of the second-order joint scattering coefficients for the chirps whose scalograms are shown in Figure \ref{fig:chirp-model}(a,b).
These coefficients do indeed show a maximum along the predicted ridge.
At low $\ell$ and high $\mu$, the approximation does not hold, but for most of the frequency range, it is accurate.
We thus see how the chirp rate $\alpha$ is captured by the joint scattering coefficients in a natural way.

\section{Audio Texture Synthesis}
\label{sec:recon}

Section \ref{sec:tf-loss} showed how mel-spectrograms and time scattering transforms do not adequately capture time-frequency structure.
As $T$ increases, this problem becomes more serious, necessitating the introduction of the joint time-frequency scattering transform.
In this section, we illustrate the representational power of this transform using texture synthesis experiments.

With the aim of generating realistic soundtracks of arbitrary duration, audio texture synthesis has many applications in virtual reality and multimedia design \cite{Schwarz2011}.
In computational neuroscience, it also offers a testbed for the comparative evaluation of biologically plausible models for auditory perception \cite{mcdermott}.
Given a signal $\x(t)$ and a time-shift invariant representation $\Phi \x$ of $\x$, the texture synthesis problem may be formulated as the minimization of the error
$E(\y) = \left\Vert \Phi \y - \Phi \x \right\Vert ^2$
between $\Phi \x$ and the representation $\Phi \y$ of the synthesized signal $\y(t)$.
Here, $\Phi$ can be a scattering transform $\tfS$, a mel-spectrogram $\MEL$, or some other representation.
Note that minimizing $E(\y)$ does not imply that $\y(t)$ approximates $\x(t)$ in any way; since $\Phi$ is a time-shift invariant representation, this is not possible.
Instead, we expect $\y(t)$ to contain examples of the time-frequency structures captured in $\Phi(\x)$.

The state of the art in the domain is held by McDermott and Simoncelli \cite{mcdermott}, who define $\Phi$ using a set of summary statistics.
These statistics are similar to the time scattering transform as they are calculated using cascades of constant-$Q$ filterbanks and pointwise nonlinearities.
However, unlike the scattering transform, which simply averages in time, McDermott and Simoncelli also compute higher-order statistical moments: variance, skewness, kurtosis, and correlation coefficients across frequency bands.
These coefficients are very sensitive to outliers in the data, which reduces their applicability to classification.

To synthesize $\y(t)$, we first initialize using random Gaussian noise with power spectral density matching the first-order scattering coefficients $\tmS_1 \x(t, \la)$ of the target waveform $\x(t)$, since these coefficients are present in all the considered representations.
We then iteratively refine the signal by gradient descent \cite{bruna-texture}.
Because the modulus nonlinearity is not convex, the error $E(\y)$ is not convex; consequently, gradient descent only converges to a local minimum of $E(\y)$.
However, this local minimum is typically of low error, with $E(\y)$ equal to around $0.02 \times \| \Phi x\|^2$ for typical audio recordings.
We found empirically that the convergence rate is increased using a fixed momentum term and a ``bold driver'' learning rate policy \cite{sutskever2013icml}.

Gradient descent in a scattering network can be implemented by backpropagation from deeper to shallower layers.
Like in a deep convolutional network, the gradient backpropagation of the convolution with each wavelet $\gu_k(t)$ corresponds to a convolution with the adjoint filter $\gu^\dagger_k (t) = \bar{\gu}_k (-t)$, obtained by time reversal and complex conjugation of $\gu_k(t)$.

\begin{figure}[!ht]
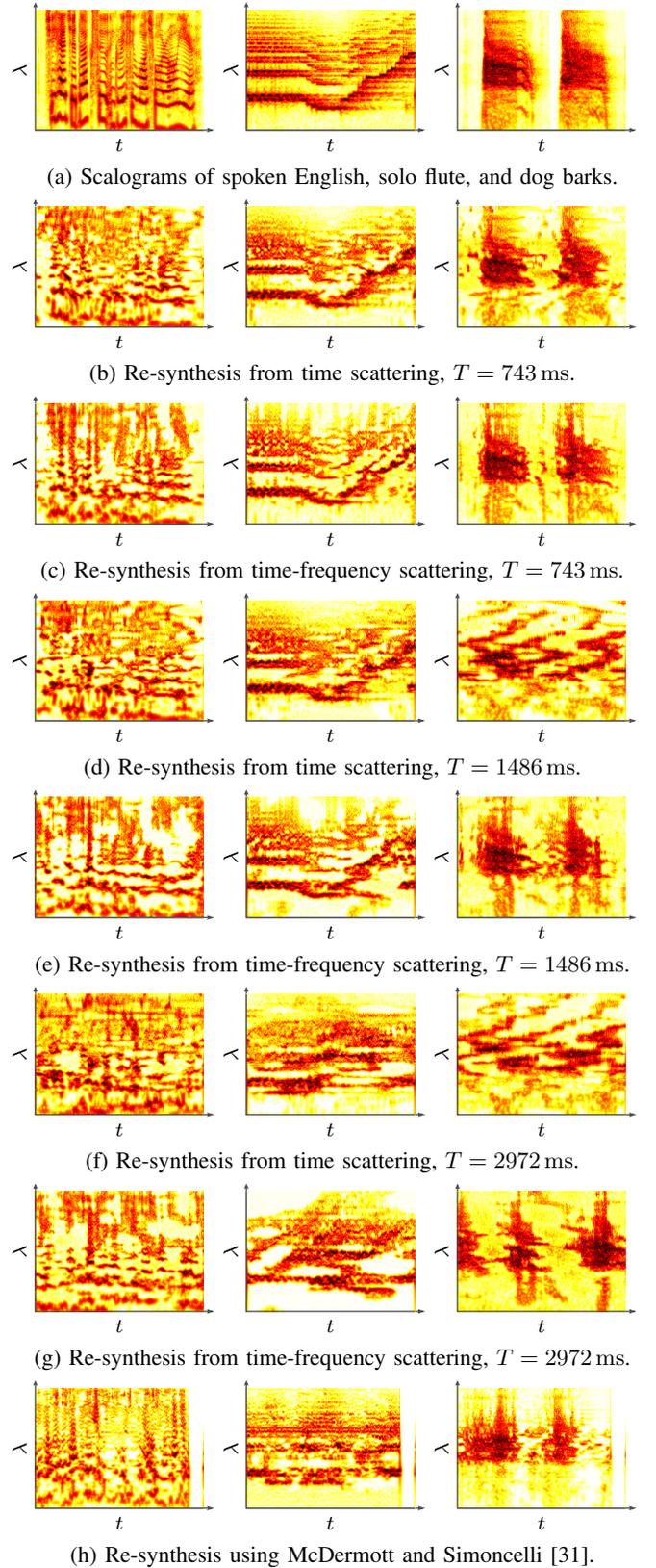

\centering
\begin{subfigure}{\linewidth}
\vspace{-0.1cm}
\input{speech_original_gplt}%
\input{flute_original_gplt}%
\input{dog-bark_original_gplt}
\caption{Scalograms of spoken English, solo flute, and dog barks.}
\end{subfigure}
\begin{subfigure}{\linewidth}
\input{speech_Q=08_J=14_sc=time_wvlt=morlet_it=100_gplt}%
\input{flute_Q=08_J=14_sc=time_wvlt=morlet_it=100_gplt}%
\input{dog-bark_Q=08_J=14_sc=time_wvlt=morlet_it=100_gplt}
\caption{Re-synthesis from time scattering, $T=\SI{743}{ms}$.}
\end{subfigure}
\begin{subfigure}{\linewidth}
\input{speech_Q=08_J=14_sc=time-frequency_wvlt=morlet_it=100_gplt}%
\input{flute_Q=08_J=14_sc=time-frequency_wvlt=morlet_it=100_gplt}%
\input{dog-bark_Q=08_J=14_sc=time-frequency_wvlt=morlet_it=100_gplt}
\caption{Re-synthesis from time-frequency scattering, $T=\SI{743}{ms}$.}
\end{subfigure}
\begin{subfigure}{\linewidth}
\input{speech_Q=08_J=15_sc=time_wvlt=morlet_it=100_gplt}%
\input{flute_Q=08_J=15_sc=time_wvlt=morlet_it=100_gplt}%
\input{dog-bark_Q=08_J=15_sc=time_wvlt=morlet_it=100_gplt}
\caption{Re-synthesis from time scattering, $T=\SI{1486}{ms}$.}
\end{subfigure}
\begin{subfigure}{\linewidth}
\input{speech_Q=08_J=15_sc=time-frequency_wvlt=morlet_it=100_gplt}%
\input{flute_Q=08_J=15_sc=time-frequency_wvlt=morlet_it=100_gplt}%
\input{dog-bark_Q=08_J=15_sc=time-frequency_wvlt=morlet_it=100_gplt}
\caption{Re-synthesis from time-frequency scattering, $T=\SI{1486}{ms}$.}
\end{subfigure}
\begin{subfigure}{\linewidth}
\input{speech_Q=08_J=16_sc=time_wvlt=morlet_it=100_gplt}%
\input{flute_Q=08_J=16_sc=time_wvlt=morlet_it=100_gplt}%
\input{dog-bark_Q=08_J=16_sc=time_wvlt=morlet_it=100_gplt}
\caption{Re-synthesis from time scattering, $T=\SI{2972}{ms}$.}
\end{subfigure}
\begin{subfigure}{\linewidth}
\input{speech_Q=08_J=16_sc=time-frequency_wvlt=morlet_it=100_gplt}%
\input{flute_Q=08_J=16_sc=time-frequency_wvlt=morlet_it=100_gplt}%
\input{dog-bark_Q=08_J=16_sc=time-frequency_wvlt=morlet_it=100_gplt}
\caption{Re-synthesis from time-frequency scattering, $T=\SI{2972}{ms}$.}
\end{subfigure}
\begin{subfigure}{\linewidth}
\input{speech_mcdermott_gplt}%
\input{flute_mcdermott_gplt}%
\input{dog-bark_mcdermott_gplt}
\caption{Re-synthesis using McDermott and Simoncelli \cite{mcdermott}.}
\end{subfigure}
\caption{
\label{fig:resynthesis}
Scalograms of audio re-synthesis using time scattering, time-frequency scattering, and McDermott and Simoncelli \cite{mcdermott}. Synthesis is performed at various time scales $T$ and inputs: spoken English (left), solo flute (middle), and dog barks (right).
}
\end{figure}

Figure \ref{fig:resynthesis} shows the synthesized scalograms of three sounds for various values of $T$.
Here, time-frequency scattering outperforms time scattering for $T$ greater than $\SI{1}{s}$.
Again, we do not expect these synthesized signals to reproduce the originals in the top row due to the imposed time-shift invariance 
In particular, speech is more intelligible due to better reconstruction of articulations, individual notes in a musical scale are more salient, and broadband impulses such as dog barks keep their typical amplitude envelopes and inter-onset intervals.
Compared to the representation of McDermott and Simoncelli \cite{mcdermott}, time-frequency scattering achieves similar quality, but does not have the same sensitivity to outliers.
Indeed, the contractivity of the wavelet transform and the modulus ensures the scattering transform's robustness to additive noise \cite{stephane,dss}.

\section{Supervised Classification}
\label{sec:classif}

We evaluate the performance of the joint time-frequency scattering transform on various classification tasks.
It is shown to enjoy significantly greater accuracy compared to baseline MFCC and time scattering approaches.
In fact, the proposed transform performs comparably to state-of-the-art learned convolutional networks whose training requires significant computational resources and large training sets.
As a result, the joint scattering transform provides a good alternative when such an expensive training step is infeasible or undesirable.

\subsection{Frequency Transposition Invariance}
\label{sec:transp}

In addition to time-shifting and time-warping, signals are also transformed by frequency-shifting and frequency-warping.
Frequency-shifting, also known as frequency transposition, changes the pitch but leaves subband envelopes intact.
This shifts the scalogram $\X(t, \la)$ by a fixed amount $\eta$ in log-frequency, giving $\X(t, \la-\eta)$.
While certain tasks are sensitive to pitch, like speaker identification, others, like speech recognition in non-tonal languages, require invariance to transposition.

The time scattering transform is rendered transposition invariant and stable to frequency-warping by applying a second scattering transform along log-frequency $\la$.
The result is the separable time and frequency scattering transform, introduced in And\'en and Mallat \cite{dss}.
Note that we may skip the averaging step of this second scattering transform.
Indeed, the averaging step is a linear map that can be learned by the classifier given enough training data \cite{dss}.

To render the joint time-frequency scattering transform transposition invariant, we similarly apply a second scattering transform along $\la$ for the first-order coefficients $\tfS_1 \x$.
For the second order $\tfS_2 \x$, however, we simply average along $\la$, since the two-dimensional wavelet decomposition already captures the relevant frequency structure.
The resulting representation then has the necessary transposition invariance and frequency-warping stability properties.
Again, if the training set is large enough, the final averaging steps can be learned by the classifier.

\subsection{Phone Segment Classification}
\label{sec:timit}

An individual phone in speech is short, on average $40~\mathrm{ms}$ in duration.
For phone identification, we therefore require the invariance scale $T$ to be of this order.
Since $T$ is small, there is less room for the type of misalignment seen in Section \ref{sec:tf-loss}.
We therefore expect the joint time-frequency scattering to provide only limited improvement over time scattering.

To evaluate, we use the TIMIT dataset, which contains recordings of spoken phrases, each labeled with its constituent phones and their locations \cite{timit}.
Given a phone segment, we wish to classify the phone according to the standard protocol \cite{kflee,clarkson-moreno}.
This task is simpler than continuous speech recognition, but provides a good framework for evaluating representations.
The training and evaluation sets consist of $3696$ and $192$ phrases, respectively.
We use a $400$--phrase validation set to optimize hyperparameters (see And\'{e}n and Mallat \cite{dss}).

Instead of the raw scattering transform, we use their logarithm, known as the log-scattering transform, as input to the classifier \cite{dss}.
We compute these coefficients over $192~\mathrm{ms}$ intervals centered on each segment with $T = 32~\mathrm{ms}$.
All coefficients are concatenated into a single vector together with the logarithm of the segment duration.
This vector is then used for classification.
The same processing is also performed for separable time and frequency scattering as well as joint time-frequency scattering.
We set the maximum frequency scale $F$ to $4$ octaves.
As a baseline, we compute Delta-MFCCs, which supplement standard MFCCs with first and second time derivatives \cite{furui}.
These are computed with the same windows and concatenation as the log-scattering coefficients.

For each representation, we train a support vector machine (SVM) \cite{cortes-vapnik} with a Gaussian kernel.
Here and in the following, we use a modified implementation of the LIBSVM library \cite{libsvm}.

\begin{table}
\begin{center}
\begin{tabular}{|l|c|}
\hline
Representation & Error (\%) \\
\hline
Delta-MFCCs & $18.3$ \\
State of the art \cite{ratajczak2017frame} & $11.9$ \\
\hline
Time scattering & $17.3$ \\
Separable time and freq.
scattering & $16.1$ \\
Joint time-freq. scattering & $15.7$ \\
\hline
\end{tabular}
\end{center}
\caption{
\label{table:timit}
Error rates for phone segment classification.
All representations are computed with $T = 32~\mathrm{ms}$ and $Q = 8$.
}
\end{table}

Results are shown in Table \ref{table:timit}.
Delta-MFCCs have an error rate of $18.3\%$, while the state-of-the-art representation, a convolutional neural network, achieves $11.9\%$ \cite{ratajczak2017frame}.
The time scattering transform obtains an error rate of $17.3\%$, which is improved by scattering along log-frequency to give $16.1\%$.
Finally, we obtain an error of $15.7\%$ for the joint time-frequency scattering transform.
As mentioned earlier, the amount of time-frequency structure present in an individual phone is small, but there is enough to give a small improvement to the joint transform.
This is partly due to the fact that certain phones (such as plosives) are characterized by their onset, which exhibits sophisticated time-frequency structure.

For this task, the joint time-frequency scattering transform does not outperform the state-of-the-art learned convolutional network.
Note, however, that the only learning involved for the scattering transform is training the SVM.
The scattering network weights are fixed, providing a simpler representation with acceptable performance.
Another important difference is that the state-of-the-art result was obtained by simultaneously estimating the labels for all phone segments in an utterance.
As a result, this network has access to more context about each segment that it can use to improve classification performance.
Combining these two approaches--a scattering transform as input to a more adaptive deep neural network--could yield even better performance as fewer parameters need to be estimated.
Indeed, replacing mel-spectrograms by scattering transforms in deep neural networks have improved performance for several tasks \cite{peddinti-dss,zeghidour,oyallon}.

\subsection{Musical instrument classification}
\label{sec:medleydb}

The timbre of a musical instrument is essentially determined by its shape and materials.
Both remain constant during a musical performance.
Therefore, musical instruments may be modeled as dynamical systems with constant parameters.
The task of musical instrument classification is to retrieve these parameters while remaining invariant to changes in pitch, intensity, and expressive technique induced by the performer.

In a musical instrument, the response of the vibrating body to an excitation is typically nonlinear.
As a result, sharp onsets produce distinctive time-frequency patterns which are not adequately captured by short-term audio descriptors operating on scales $T \approx \SI{20}{\milli\second}$, a typical window size for MFCCs.
Joint time-frequency scattering, on the other hand, captures such patterns up to the scale $T \approx \SI{3}{\second}$ of a short musical phrase.

To illustrate this, we apply it to automatic instrument classification in solo phrases with a taxonomy of eight instruments.
In line with the cross-collection methodology of Bogdanov et al. \cite{bogdanov}, we train and validate all models on the MedleyDB v1.1 dataset \cite{Bittner2014} and test them on the solosDb dataset \cite{Joder2009}.
This is the evaluation setting of Lostanlen and Cella \cite{spiralnetworks}.

Results are shown in Table \ref{table:medleydb}.
It appears that all models which do not explicitly decompose in both time and log-frequency (Delta-MFCCs, time scattering, and a convolutional network of temporal convolutions on the scalogram) perform comparably, with errors around $38\%$.
Introducing decompositions along the log-frequency axis through time-frequency convolutional networks and spiral convolutional networks, we obtain error rates of $28.3\%$ and $26.0$, respectively \cite{spiralnetworks}.
The improvement likely stems from the fact that musical instruments carry important discriminative information in the temporal evolution of their spectral envelopes as well as frequency modulation structures, both of which are captured by joint decompositions in time and log-frequency.
The joint time-frequency scattering transform further reduces the error to $22.0\%$.
The small size of the training set makes optimizing a convolutional network difficult, which may partially explain the improved accuracy of the joint scattering transform compared to the fully learned convolutional networks.

\begin{table}
\begin{center}
\begin{tabular}{|l|c|}
\hline
Representation & Error (\%) \\
\hline
Delta-MFCCs & $39.3$ \\
Time convolutional networks & $38.2$ \\
Time-frequency convolutional networks & $28.3$ \\
Spiral convolutional networks \cite{spiralnetworks} & $26.0$ \\
\hline
Time scattering  & $38.0$ \\
Time-frequency scattering & $22.0$ \\
\hline
\end{tabular}
\end{center}
\caption{\label{table:medleydb}
Error rates for musical instrument classification.
All representations are computed with $T = 3~\mathrm{s}$, $Q = 12$, and $F = 4$ octaves.
}
\end{table}

\subsection{Acoustic Scene Classification}
\label{sec:esc}

Environmental sounds and acoustic scenes are characterized by larger-scale time-frequency structures.
These recordings typically stretch over several seconds, each composed of shorter sound events which characterize the scene.
This could be birdsong in a park, car horns in a street, or the scraping of chairs in a caf\'{e}.
To differentiate between different sequences of such events, we must characterize longer-range structures.
As discussed above, this is not possible using standard representations, such as MFCCs or time scattering, which do not adequately capture time-frequency structure.

We evaluate the joint scattering transform on three acoustic scene datasets: UrbanSound8K (US8K) \cite{us8k}, ESC-50 \cite{esc50}, and DCASE2013 \cite{dcase2013}.
US8K and DCASE2013 have $10$ classes each, while ESC-50 contains $50$ classes, ranging from gun shots and subway stations to crying babies and supermarkets.
Both US8K and ESC-50 contain several thousand recordings of approximate duration $4~\mathrm{s}$.
DCASE2013, on the other hand, contains $100$ (public) training samples and $100$ (private) evaluation samples, each of duration $30~\mathrm{s}$.
All recordings being relatively long, they may exhibit sophisticated time-frequency structures that are discriminative for classification.

For US8K and ESC-50, we compute scattering transforms with $Q = 8$ and $T = 4~\mathrm{s}$.
We choose a large value for $T$ because there are long, texture-like structures in this dataset that we would like to characterize.
To ensure some transposition invariance, we explicitly average the separable and joint transforms over $F = 1$ octave (US8K) or $F = 2$ octaves (ESC-50).
Here, we do not want to choose a large frequency scale $F$ since some pitch information is necessary to distinguish certain sounds.
Since $T$ equals the clip duration, each clip yields a single scattering vector, which is fed into the classifier.

For DCASE2013, we compute scattering transforms with $Q = 4$, $T = 1.5~\mathrm{s}$, and frequency averaging over $F = 8$ octaves where applicable.
We must select parameters different from those of US8K and ESC-50 due to the much smaller size of DCASE2013.
Choosing smaller values for $Q$ and $T$ limits the complexity of the time-frequency structure captured by the transform, while choosing a large $F$ and averaging along frequency creates additional invariance to transposition.
Since $T$ is much smaller than the recording duration ($30~\mathrm{s}$), this yields multiple scattering vectors which are classified separately.
The overall class is then obtained by majority voting.

Delta-MFCCs are computed for all datasets as a baseline.
For each representation, we train a linear SVM with hyperparameters optimized by cross-validation on the training set.

The error for US8K and ESC-50 is calculated through cross-validation on pre-specified folds.
For these datasets, we use the data augmentation scheme of Salamon and Bello \cite{salamon-augment}, but without pitch-shifting, since transposition invariance is already enforced.
We calculate the DCASE2013 error on the evaluation subset in accordance with previous work \cite{aytar,arandjelovic-zisserman}.

\begin{table}
\begin{center}
\begin{tabular}{|l|c|c|c|}
\hline
Representation & US8K & ESC-50 & DCASE2013 \\
\hline
Delta-MFCCs \cite{us8k,esc50} & $46.0$ & $56.0$ & $42$ \\
Salamon and Bello \cite{salamon-augment} & $21.0$ & -- & -- \\
SoundNet \cite{aytar} & -- & $25.8$ & $12$ \\
$\mathrm{L}^3$ network \cite{arandjelovic-zisserman} & -- & $20.7$ & $7$ \\
\hline
Time scatt. & $26.9 \pm 4.1$ & $39.3 \pm 2.2$ & $12$ \\
Separable time and freq. scatt. & $22.8 \pm 3.0$ & $26.0 \pm 2.7$ & $6$ \\
Joint time-freq. scatt. & $19.6 \pm 2.9$ & $21.8 \pm 2.0$ & $5$ \\
\hline
\end{tabular}
\end{center}
\caption{
\label{table:esc}
Average and standard deviation of error rates for scene classification on US8K, ESC-50, and DCASE2013.
}
\end{table}

Results are shown in Table \ref{table:esc}.
The Delta-MFCCs have error rates of $46.0\%$, $56.0\%$, and $42\%$ for US8K, ESC-50, and DCASE2013, respectively.
State-of-the-art convolutional networks, on the other hand, obtain $21.0\%$, $20.7\%$ and $7\%$.

The standard time scattering transform yields accuracies of $26.8\%$ (US8K), $39.3\%$ (ESC-50), and $12\%$ (DCASE2013), improving on Delta-MFCCs by better capturing the temporal structure of each subband.
Adding a scattering transform along the log-frequency axis improves results to $22.8\%$ (US8K), $26.0\%$ (ESC-50), and $6\%$ (DCASE2013).
This improvement is expected since these sounds exhibit significant pitch variability which is not discriminative to each class.

The joint time-frequency scattering transform performs even better, giving errors of $19.6\%$ (US8K), $21.8\%$ (ESC-50), and $5\%$ (DCASE2013).
This is partly because environmental sounds are often characterized by dynamic filters which evolve in time, creating a spectrotemporal filter.
The mechanical and biological nature of these sounds also results in frequency modulation.
Both phenomena are examples of time-frequency geometry which are well-characterized by the joint scattering transform.
From a different perspective, the recordings in these datasets are sensitive to frequency-dependent time-shifts (see Section \ref{sec:tf-loss}).
Indeed, taking a signal with many transients, such as a jackhammer in a street scene, and misaligning its subbands yields a completely different sound.
A representation sensitive to such transformations is therefore expected to perform better.

Again, the joint scattering transform performs comparably to learned convolutional neural networks.
However, learned networks require significant computational resources to train and certain expertise in designing the network.
Both SoundNet and the $\mathrm{L}^3$ network are pretrained on large external datasets, requiring several days of computation on graphics processing units.
In contrast, the joint scattering transform has a fixed network structure, so the only training needed is for the SVM, requiring at most a few hours.
By considering the invariances of the problem (time-shifting, frequency transposition) and the structures we would like to capture (joint time-frequency geometry), we obtain good performance without costly pretraining.

\section{Conclusion}

We introduced a joint time-frequency scattering transform, a time-shift invariant descriptor with state-of-the-art classification performance for a wide range of audio datasets.
Important improvements are obtained for classification tasks involving large-scale signal structures.
Time-frequency scattering descriptors also recover complex signals including audio textures.

A joint time-frequency scattering has a computational structure similar to deep convolutional networks \cite{mallat2016understanding}, but is calculated with fixed wavelet filters.
It thus requires less training data to obtain accurate classification results.
However, when more training examples are available, learned convolutional networks provide state-of-the-art results.
Indeed, these networks adapt the representation to each classification problem.
Taking into account prior information on time-frequency geometry could help improve their performance.

\begin{appendices}
\section{}
\label{sec:appendix}

\begin{lemma}
\label{lem:chirp-wave}
Let $\tmwav_\la(t)$ be as defined in Theorem \ref{thm:fm}.
For $\x(t) = \exp(2\pi \imunit \, 2^{\alpha t})$, we then have
\begin{equation}
\label{eq:chirp-wave-approx-final}
|\x \conv \tmwav_\la|(t) = |\tmhwav \left( \log (2^\alpha) \, 2^{\alpha t-\la} \right)| + \varepsilon(t, \la)~,
\end{equation}
where
\begin{equation}
|\varepsilon(t, \la)| \le C |\alpha| 2^{-\la}
\end{equation}
for some constant $C > 0$ which only depends on $\tmwav(t)$.
\end{lemma}

\begin{proof}
If $\alpha = 0$, $\x \conv \tmwav_\la(t) = \tmhwav (0) \exp(2\pi\imunit)$.
We therefore assume that $\alpha \neq 0$.
If $\supp \tmwav \subset [-\Delta, \Delta]$, we have
\begin{equation}
\label{eq:chirp-wave-integral}
    \x \conv \tmwav_\la(t) = \smashoperator{\int_{|u| \le 2^{-\la} \Delta}} \exp(2\pi \imunit \, 2^{\alpha(t-u)}) \tmwav_\la(u)\,\diff u~.
\end{equation}
For $u$ close to zero, the derivative of $2^{\alpha(t-u)}$ is approximately $-\log(2^\alpha) 2^{\alpha t}$.
We exploit this to integrate \eqref{eq:chirp-wave-integral} by parts.
Let $\g(u) = \exp(2\pi \imunit \, 2^{\alpha t} (2^{-\alpha u} + u \log(2^\alpha) ) )$.
We then have
\begin{align}
\nonumber
&\x \conv \tmwav_\la(t) = \smashoperator{\int_{|u| \le 2^{-\la} \Delta}} \g(u) \exp(-2\pi \imunit \, u \log (2^\alpha) 2^{\alpha t} ) \tmwav_\la(u)\,\diff u \\
\label{eq:chirp-wave-approx1}
&\quad = \g(2^{-\la} \Delta) \tmhwav_\la(\log (2^\alpha) 2^{\alpha t}) - \smashoperator{\int_{|u| \le 2^{-\la} \Delta}} \g'(u) \I(u)\,\diff u~,
\end{align}
where $\I(u) = \int_{-2^{-\la} \Delta}^u \exp(-2\pi \imunit \, v \log(2^\alpha) 2^{\alpha t}) \tmwav_\la(v)\,\diff v$.

The magnitude of the second term in \eqref{eq:chirp-wave-approx1} is bounded by
\begin{equation}
\label{eq:chirp-wave-err1}
2\pi \, 2^{\alpha t} |\log(2^{\alpha})| (1-2^{-|\alpha| 2^{-\la} \Delta}) 2^{1-\la} \Delta \max_{|u| \le 2^{-\la} \Delta} |\I(u)|~.
\end{equation}
Integrating $\I(u)$ by parts and taking the modulus gives
\begin{equation*}
|\I(u)| \le \frac{\|\tmwav_\la\|_\infty + \|\tmwav_\la'\|_1}{2\pi |\log(2^\alpha)|\,2^{\alpha t}} = 2^\la\frac{\|\tmwav\|_\infty + \|\tmwav'\|_1}{2\pi |\log (2^\alpha)|\, 2^{\alpha t}}~,
\end{equation*}
since $\alpha \neq 0$, $\|\tmwav_\la\|_\infty = 2^\la \|\tmwav\|_\infty$ and $\|\tmwav_\la'\|_1 = 2^\la \|\tmwav'\|_1$.
Plugging this into \eqref{eq:chirp-wave-err1}, we obtain
\begin{equation}
2 \Delta (1 - 2^{-|\alpha| 2^{-\la} \Delta}) (\|\tmwav\|_\infty + \|\tmwav'\|_1)~.
\end{equation}
Given that $1 - 2^{-u} < \log(2) u$ for all $u > 0$, this simplifies to
\begin{equation}
\label{eq:chirp-wave-err2}
2 \log(2) \Delta^2 (\|\tmwav\|_\infty + \|\tmwav'\|_1) \, |\alpha| 2^{-\la}~.
\end{equation}

Since $|\g(u)| = 1$, the modulus of the first term in \eqref{eq:chirp-wave-approx1} is $|\tmhwav_\la(\log (2^\alpha) 2^{\alpha t} )| = |\tmhwav(\log (2^\alpha) 2^{\alpha t - \la})|$.
The triangle inequality then establishes \eqref{eq:chirp-wave-approx-final} with $|\varepsilon(t, \la)|$ bounded by \eqref{eq:chirp-wave-err2}.
\end{proof}

\begin{lemma}
\label{lem:scal-decomp}
Define $\tmwav_\la(t)$, $\tfwav_{\mu,\ell,s}(t, \la)$, and $c_0$ as in Theorem \ref{thm:fm} and let
\begin{equation*}
t_0(\la) = \frac{\la}{\alpha} - \frac{\log\log 2^\alpha}{\log 2^\alpha}.
\end{equation*}
Given
\begin{equation*}
\Y(t, \la) = |\tmhwav\left( \log (2^\alpha)\,2^{\alpha t - \la} \right)|~,
\end{equation*}
its two-dimensional wavelet modulus decomposition satisfies
\begin{align}
\label{eq:scal-decomp-approx-final}
&|\Y \conv \tfwav_{\mu,\ell,s}| (t, \la) \\
\nonumber
&\quad = \frac{c_0|\tfwavtm_\mu|( t - t_0(\la) )}{\alpha} \left|\tfhwavfr\left(-\frac{s 2^{\mu-\ell}}{\alpha} \right)\right| + \varepsilon(t, \la, \mu, \ell, s)~,
\end{align}
where $|\tfwavtm_\mu|(t) = 2^\mu |\tfwavtm|(2^\mu t)$, and
\begin{equation*}
|\varepsilon(t, \la, \mu, \ell, s)| \le C(2^{2\mu} |\alpha| ^{-2} + 2^{2\mu-\ell} |\alpha|^{-2})~,
\end{equation*}
for some $C > 0$ depending only on $\tmwav(t)$, $\tfwavtm(t)$, $\tfwavfr(\la)$.
\end{lemma}

\begin{proof}
Since $|\widehat \tmwav (\om) |$ is maximized at $\om = 1$, fixing $\la$, the maximum of $\Y(t, \la)$ is at $t_0(\la)$.
For small enough $\mu$, $\Y(t, \la)$ approximates a Dirac delta function centered at $t_0(\la)$.
We exploit this when convolving $\Y(t, \la)$ by $\tfwavtm_\mu(t)$.

Approximating $\tfwavtm(u)$ with its value at $u = t - t_0(\la)$ gives
\begin{align*}
&\Y(\cdot, \la) \conv \tfwavtm_\mu(t) = \int_\R \Y(t-u, \la) \tfwavtm_\mu(u)\,\diff u \\
&\quad = \int_\R | \widehat \tmwav ( \log (2^\alpha) 2^{\alpha (t-u) - \la} ) | \, \times \\
&\qquad (\tfwavtm_\mu(t-t_0(\la)) + \varepsilon_1(t, \la, \mu, u))\,\diff u~,
\end{align*}
where $|\varepsilon_1(t, \la, \mu, u)| \le |t-t_0(\la)-u| \, \|{\tfwavtm_\mu}'\|_\infty$.
Setting $\varepsilon_2(t, \la, \mu) = \int_\R | \widehat \tmwav ( \log (2^\alpha) 2^{\alpha (t-u) - \la} ) | \varepsilon_1(t, \la, \mu, u)\,\diff u$ gives
\begin{equation}
\label{eq:scal-decomp-approx1}
\Y(\cdot, \la) \conv \tfwavtm_\mu(t) = c_0 \alpha^{-1} \tfwavtm_\mu(t-t_0(\la)) + \varepsilon_2(t, \la, \mu)~,
\end{equation}
using a change of variables, where $c_0 = \int_\R |\widehat \tmwav(2^u)|\,\diff u < \infty$.

We bound $\varepsilon_2(t, \la, \mu)$ through
\begin{align*}
&|\varepsilon_2(t, \la, \mu)| \le \int_\R | \widehat \tmwav ( \log (2^\alpha) 2^{\alpha (t-u) - \la} ) | \, |\varepsilon_1(t, \la, \mu, u)|\,\diff u \\
\nonumber
&\quad \le \|{\tfwavtm_\mu}'\|_\infty \int_\R | \widehat \tmwav ( \log (2^\alpha) 2^{\alpha (t-u) - \la} ) | \, |t-t_0(\la)-u|\,\diff u~.
\end{align*}
The change of variables $t-t_0(\la)-u \mapsto \alpha^{-1} u$ now gives
\begin{equation}
\label{eq:scal-decomp-err1}
|\varepsilon_2(t, \la, \mu)| \le |\alpha|^{-2} 2^{2\mu} \|{\tfwavtm}'\|_\infty \int_\R | \tmhwav (2^u) |\, |u|\,\diff u,
\end{equation}
where we have used $\|{\tfwavtm_\mu}'\|_\infty = 2^{2\mu} \|{\tfwavtm}'\|_\infty$.

We now convolve \eqref{eq:scal-decomp-approx1} by $\tfwavfr_{\ell,s}(\la) = 2^\ell \tfwavfr(s2^\ell \la)$.
At high $\ell$, this wavelet will mostly capture phase variation.
To see this, we factorize $\tfwavtm_\mu(t)$ into an envelope and a phase, yielding $|\tfwavtm_\mu|(t) \exp(2\pi \imunit \, 2^\mu t))$.
The convolution then becomes
\begin{align}
\nonumber
& c_0 \alpha^{-1} \tfwavtm_\mu(t-t_0(\cdot)) \conv \tfwavfr_{\ell,s}(\la) \\
\nonumber
& \quad = c_0\alpha^{-1} \int_\R |\tfwavtm_\mu|(t-t_0(\la-\gamma)) \, \times \\
\label{eq:scal-decomp-approx11}
& \quad\quad \exp(2\pi \imunit \, 2^\mu(t-t_0(\la-\gamma))) \tfwavfr_{\ell,s}(\gamma)\,\diff\gamma~.
\end{align}
We now make the approximation
\begin{equation*}
|\tfwavtm_\mu|(t-t_0(\la-\gamma)) = |\tfwavtm_\mu|(t-t_0(\la)) + \varepsilon_3(t, \la, \mu, \gamma)~,
\end{equation*}
where $|\varepsilon_3(t, \la, \mu, \gamma)| \le \| |\tfwavtm_\mu|' \|_\infty |t_0(\la-\gamma)-t_0(\la)|$.
Plugging this into \eqref{eq:scal-decomp-approx11}, we obtain
\begin{align}
\label{eq:scal-decomp-approx12}
& c_0 \alpha^{-1} |\tfwavtm_\mu|(t-t_0(\la)) \, \times \\
\nonumber
& \quad \int_\R \exp(2\pi \imunit \, 2^\mu(t-t_0(\la-\gamma))) \tfwavfr_{\ell,s}(\gamma)\,\diff\gamma + \varepsilon_4(t, \la, \mu, \ell, s)~,
\end{align}
where $\varepsilon_4(t, \la, \mu, \ell, s)$ equals
\begin{equation*}
c_0 \alpha^{-1} \int_\R \varepsilon_3(t, \la, \mu, \gamma) \exp(2\pi \imunit \, 2^\mu(t-t_0(\la-\gamma))) \tfwavfr_{\ell,s}(\gamma)\,\diff\gamma~.
\end{equation*}
Since $t_0(\la-\gamma) = t_0(\la) - \alpha^{-1}\gamma$, the first term in \eqref{eq:scal-decomp-approx12} is
\begin{equation}
\label{eq:scal-decomp-approx2}
c_0 \alpha^{-1} |\tfwavtm_\mu|(t-t_0(\la)) \euler^{\imunit 2^\mu(t-t_0(\la))} \tfhwavfr_{\ell,s}(-2^\mu \alpha^{-1})~.
\end{equation}

The same property of $t_0(\la)$ lets us bound $\varepsilon_4(t, \la, \mu, \ell, s)$ by
\begin{align}
\nonumber
&|\varepsilon_4(t, \la, \mu, \ell, s)| \le c_0 |\alpha|^{-2} \| |\tfwavtm_\mu|' \|_\infty \int_\R |\gamma| |\tfwavfr_{\ell,s}(\gamma)|\,\diff\gamma \\
\label{eq:scal-decomp-err2}
&\quad = c_0 |\alpha|^{-2} 2^{2\mu-\ell} \| |\tfwavtm|' \|_\infty \int_\R |\gamma| |\tfwavfr(\gamma)|\,\diff\gamma~,
\end{align}
which follows from change of variables and from $\| |\tfwavtm_\mu|' \|_\infty = 2^{2\mu} \| |\tfwavtm| '\|_\infty$.
We must also convolve $\varepsilon_2(t, \la, \mu)$ with $\tfwavfr_{\ell,s}(\la)$.
Since $\|\tfwavfr_{\ell,s}\|_1 = \|\tfwavfr\|_1$ for all $\ell, s$, we have
\begin{equation}
\label{eq:scal-decomp-err3}
|\varepsilon_2(t, \cdot, \mu) \conv \tfwavfr_{\ell,s}(\la)| \le \|\varepsilon_2(t, \cdot, \mu)\|_\infty \, \|\tfwavfr\|_1~.
\end{equation}

Combining \eqref{eq:scal-decomp-approx1} with \eqref{eq:scal-decomp-approx2} and taking the modulus yields \eqref{eq:scal-decomp-approx-final} since $\tfhwavfr_{\ell,s}(\om) = \tfhwavfr(s 2^{-\ell} \om)$, where the bound on $\varepsilon(t, \la, \mu, \ell, s)$ follows from \eqref{eq:scal-decomp-err1}, \eqref{eq:scal-decomp-err2}, \eqref{eq:scal-decomp-err3}, and the triangle inequality.
\end{proof}

\begin{proof}[Proof of Theorem \ref{thm:fm}]
Lemma \ref{lem:chirp-wave} gives
\begin{equation*}
\X(t, \la) = |x \conv \tmwav_\la|(t) = |\widehat \tmwav( \log (2^\alpha) 2^{\alpha t - \la} )| + \varepsilon_1(t, \la)~,
\end{equation*}
where $|\varepsilon_1(t, \la)| \le C_1 |\alpha| 2^{-\la}$ for some $C_1 > 0$.
We now convolve $\X(t, \la)$ with $\tfwavtm_\mu(t)$ in time $\tfwavfr_{\ell,s}(\la)$ in log-frequency and take the modulus.
Lemma \ref{lem:scal-decomp} approximates the convolution of the first term.
For the second term, we observe that
\begin{align*}
&\|\varepsilon_1(\cdot, \la) \conv \tfwavtm_\mu\|_\infty \le \|\varepsilon_1(\cdot, \la)\|_\infty \|\tfwavtm\|_1 \le C_2 |\alpha| 2^{-\la}~,
\end{align*}
for some $C_2 > 0$, since $\|\tfwavtm_\mu\|_1 = \|\tfwavtm\|_1$ for all $\mu$.
Now,
\begin{align*}
&\left| \varepsilon_1 \conv \tfwav_{\mu,\ell,s}(t, \la) \right| \le C_2 |\alpha| \int_\R 2^{-(\la-\mu)} |\tfwavfr_{\ell,s}(\mu)|\,\diff\mu \\
&\quad = C_2 |\alpha| 2^{-\la} \int_{|\mu| \le 2^{-\ell} A} 2^\mu |\tfwavfr_{\ell,s}(\mu)|\,\diff\mu \\
&\quad \le C_2 |\alpha| 2^{-\la} 2^{2^{-\ell} A} \|\tfwavfr_{\ell,s}\|_1 = C_3 |\alpha| 2^{-\la+2^{-\ell}A}~,
\end{align*}
for some $C_3 > 0$, since $\tfwavfr$ is supported on $[-A, A]$.

As a result,
\begin{align}
\label{eq:fm-wavemod-approx}
&|\X \conv \tfwav_{\mu,\ell,s}(t, \la)| \\
\nonumber
&\quad = \frac{c_0}{\alpha} |\tfwavtm_\mu|(t-t_0(\la)) \left|\tfhwavfr\left(-\frac{s 2^{\mu-\ell}}{\alpha}\right)\right| + \varepsilon_2(t, \la, \mu, \ell, s)~,
\end{align}
where
\begin{equation*}
|\varepsilon_2(t, \la, \mu, \ell, s)| \le C(|\alpha| 2^{-\la+2^{-\ell}A} + |\alpha|^{-2} 2^{2\mu} + |\alpha|^{-2} 2^{2\mu-\ell})~.
\end{equation*}
Since this bound is constant in $t$ and $\|\tmlow_T\|_1 = \|\tmlow\|_1$ for all $T$, it still holds after convolving \eqref{eq:fm-wavemod-approx} with $\tmlow_T(t)$ .
\end{proof}

\end{appendices}

\section*{Acknowledgments}

The authors would like to thank J. Salamon for sharing his data augmentation code and A. Barnett for helpful discussions on oscillatory integrals.

\bibliographystyle{IEEEtran}
\bibliography{refs}

\end{document}